\documentclass[11pt]{article}

\usepackage{setspace}
\usepackage{lscape}
\usepackage{pdflscape}

\usepackage{amsmath,amssymb,amsthm}

\usepackage{wrapfig}

\usepackage{mathrsfs}
\usepackage{url}
\usepackage{latexsym}
\usepackage{bm} 

\usepackage{mathtools}

\usepackage{booktabs,multirow}
\usepackage{array}
\newcolumntype{P}[1]{>{\centering\arraybackslash}p{#1}}
\usepackage{diagbox}

\usepackage{algorithm}

\usepackage{thmtools,thm-restate}

\usepackage{hyperref}
\usepackage{cleveref}
\expandafter\def\csname ver@etex.sty\endcsname{3000/12/31}

\usepackage{autonum}
\makeatletter
\autonum@generatePatchedReferenceCSL{Cref}
\makeatother

\usepackage{tikz,tikz-3dplot}
\usetikzlibrary{positioning,shapes,shapes.callouts,shapes.multipart,arrows.meta,fit,calc}
\usepackage{pgfplots}
\pgfplotsset{compat=newest}
\usepackage{pxpgfmark}
\usetikzlibrary{bayesnet}

\usepackage{multicol}
\usepackage{wasysym}
\usepackage{adjustbox}
\usepackage{xparse}
\usepackage{pbox}
\usepackage{authblk}
\usepackage{xspace}

\usepackage[subrefformat=parens]{subcaption}
\usepackage[inline]{enumitem}
\hypersetup{
  colorlinks,
  linkcolor={green!35!black},
  citecolor={green!35!black},
  urlcolor={blue!50!black}
}

\usepackage{tabularray}

\DeclareMathOperator*{\argmax}{argmax}
\DeclareMathOperator*{\argmin}{argmin}

\DeclareMathOperator{\EE}{\mathbb{E}}

\DeclarePairedDelimiterX\inner[2]{\langle}{\rangle}{{#1},{#2}}
\DeclarePairedDelimiter\abs{|}{|}

\DeclarePairedDelimiter\set{\{}{\}}
\DeclarePairedDelimiter\prn{(}{)}
\DeclarePairedDelimiter\bra{[}{]}
\DeclarePairedDelimiterX\Set[2]{\{}{\}}{\mspace{2mu}{#1}\;\delimsize|\;{#2}\mspace{2mu}}
\DeclarePairedDelimiterX\Prn[2]{(}{)}{\mspace{2mu}{#1}\;\delimsize|\;{#2}\mspace{2mu}}
\DeclarePairedDelimiterX\Bra[2]{[}{]}{\mspace{2mu}{#1}\;\delimsize|\;{#2}\mspace{2mu}}

\newcommand{\R}{\mathbb R}

\renewcommand{\epsilon}{\varepsilon}

\NewDocumentCommand{\exsub}{s m O{} m}{%
  \IfBooleanT{#1}{\EE_{#2}\nolimits\bra*{#4}}%
  \IfBooleanF{#1}{\EE_{#2}\nolimits\bra[#3]{#4}}%
}

\newcommand{\memo}[1]{{\color{orange}}}

\newcommand{\dgamma}[2]{\gamma_{#1}(#2)}
\newcommand{\dGamma}[2]{\Gamma_{#1}(#2)}

\newcommand{\diff}[1]{{#1}'}

\newcommand{\mt}{\tilde{t}}

\newcommand{\sol}{0}

\declaretheoremstyle[
shaded={bgcolor=gray!15},
]{thmsty}

\declaretheorem[
  name=Theorem,
  refname={Theorem,Theorems},
  style=thmsty,
]{theorem}

\declaretheorem[
  name=Lemma,
  refname={Lemma,Lemmas},
  style=thmsty,
  sibling=theorem,
]{lemma}

\crefname{algorithm}{Algorithm}{Algorithms}
\crefname{line}{Line}{Lines}
\crefname{section}{Section}{Sections}
\crefname{appendix}{Appendix}{Appendices}
\crefname{table}{Table}{Tables}
\crefname{figure}{Figure}{Figures}
\crefname{equation}{}{}
\Crefname{equation}{Eq.}{Eqs.}

\setlist[itemize]{
  topsep=0.4\baselineskip,
  itemsep=0\baselineskip,
  leftmargin=1.5em,
}

\setlist[enumerate]{
  font=\upshape,
  label=(\alph*),
  ref=(\alph*),
  topsep=0.4\baselineskip,
  itemsep=0\baselineskip,
  leftmargin=2em,
}

\newlist{enuminasm}{enumerate}{1}
\setlist[enuminasm]{
  font=\upshape,
  label=(\alph*),
  ref=\theassumption(\alph*),
  topsep=0.4\baselineskip,
  itemsep=0\baselineskip,
  leftmargin=2em,
}
\crefalias{enuminasmi}{assumption}

\newlist{enuminthm}{enumerate}{1}
\setlist[enuminthm]{
  font=\upshape,
  label=(\alph*),
  ref=\thetheorem(\alph*),
  topsep=0.4\baselineskip,
  itemsep=0\baselineskip,
  leftmargin=2em,
}
\crefalias{enuminthmi}{theorem}

\newlist{enuminlem}{enumerate}{1}
\setlist[enuminlem]{
  font=\upshape,
  label=(\alph*),
  ref=\thelemma(\alph*),
  topsep=0.4\baselineskip,
  itemsep=0\baselineskip,
  leftmargin=2em,
}
\crefalias{enuminlemi}{lemma}

\usepackage[numbers, sort]{natbib}

\newcommand{\email}[1]{\href{mailto:#1}{\nolinkurl{#1}}}

\author[1]{Yasunori Akagi\footnote{Corresponding author. E-mail: \email{yasunori.akagi@ntt.com}}}
\author[2]{Hideaki Kim}
\author[3]{Daichi Fushihara}
\author[2]{Ryosuke Nakahama}
\author[2]{Hiroyasu Miyazaki}
\author[1]{Takeshi Kurashima}

\affil[1]{NTT Human Informatics Laboratories, Kanagawa, Japan}
\affil[2]{NTT Communication Science Laboratories, Kyoto, Japan}
\affil[3]{NTT Network Service Systems Laboratories, Tokyo, Japan}

\title{A Tractable Continuous-Time Model for Designing Interventions for Time-Inconsistent Agents}

\begin{document}
\maketitle

\begin{abstract}
Designing effective goals and rewards for time-inconsistent agents is a central problem in many long-term tasks, such as learning, exercise, work, and project completion. 
An agent may initially plan to complete a task, but later abandon it because, under non-exponential discounting, the perceived trade-off between immediate effort and delayed reward changes over time.
This paper develops a tractable continuous-time model for analyzing and designing interventions for such agents in deadline-constrained progress-based tasks. In the model, an agent repeatedly chooses a future progress trajectory that minimizes perceived cost and then follows its infinitesimal initial direction. Although this leads to a continuous-time dynamic behavior defined through a variational problem, we show that the resulting trajectory admits a concise analytical representation under generalized hyperbolic discounting, a broad class of discount functions that includes exponential and hyperbolic discounting as special cases. Using this representation, we characterize when the agent completes the task, abandons it immediately, or exhibits time-inconsistent abandonment after making partial progress. We then study two intervention design problems: optimal goal setting and optimal reward scheduling. For goal setting, we derive optimal goals both when exploitative rewards are allowed and when they are prohibited, and we identify conditions under which exploitative rewards are ineffective. For reward scheduling, we show that, for a fixed number of stages, equal-length periods and equal rewards are optimal, and that finer reward splitting monotonically improves final progress up to a discount-independent limit. These results provide a continuous-time framework for intervention design for time-inconsistent agents and clarify how optimal interventions differ from those in existing discrete-time models.

\end{abstract}

\section{Introduction}
Many long-term tasks require an external designer to choose goals, deadlines, and rewards for agents whose effort unfolds over time. Examples include a learning platform that sets milestones for a student, a company that offers bonuses for reaching sales targets, a health program that rewards exercise, or a project manager who divides a long project into intermediate deliverables. In such settings, the designer is not merely interested in predicting whether the agent will complete the task; the designer must also decide how to structure the task so that the agent makes as much progress as possible. This raises a basic intervention design problem: how should goals and rewards be chosen for agents who may fail to act according to their own long-term plans?

A central reason why such failures occur is that agents may evaluate present and future costs differently over time. This is usually described through \emph{time preferences}, or equivalently through a \emph{time discount function}, which specifies how an agent discounts costs and rewards that occur in the future \cite{frederick2002time}. In the standard exponential-discounting model, the relative value of two future outcomes is preserved as time passes, and the resulting preferences are dynamically consistent \cite{samuelson1937note}. By contrast, non-exponential discounting generally leads to \emph{time inconsistency}: a plan that appears optimal from an earlier viewpoint may no longer appear optimal at a later time \cite{strotz1955myopia,frederick2002time}. One of the most prominent sources of such behavior is \emph{present bias}, in which immediate costs and rewards receive disproportionately large weight relative to delayed ones \cite{phelps1968second,laibson1997golden}.

This phenomenon is especially important in deadline-constrained progress-based tasks. An agent may initially plan to complete a task because the delayed reward appears worth the required effort, but later abandon the task when the immediate cost of effort becomes more salient. For example, a person may plan to exercise regularly for a month in order to receive a health-related reward, but later skip training sessions as the immediate effort becomes more burdensome than it appeared at the planning stage. Similar patterns arise in studying, work, writing, and other cumulative-effort tasks. From the viewpoint of an intervention designer, the key question is not only why such abandonment occurs, but also how goals, deadlines, and rewards can be designed to mitigate it.

Time-inconsistent planning has therefore become an important topic at the interface of artificial intelligence, economics, and computation. A seminal computational model was introduced by \citet{kleinberg2014time}, who represented a task as a graph and modeled a present-biased agent as repeatedly choosing a path with minimum perceived cost. This model captures qualitative phenomena such as procrastination and abandonment, and has stimulated a line of work on interventions for time-inconsistent agents, including motivating subgraphs, penalties, and other forms of task design \cite{albers2019motivating,albers2021value,tang2017computational}. However, the general graph-based formulation is often computationally difficult, and many natural intervention design problems are NP-hard.

To obtain stronger tractability, Akagi et al. \cite{akagi2024analytically,akagi2026delta} studied \emph{progress-based tasks}. In a progress-based task, an agent accumulates a real-valued measure of progress over a time horizon and receives a reward if the progress reaches a prescribed goal. This class does not cover all sequential decision problems, but it captures an important family of cumulative-effort tasks, such as completing a thesis, accumulating study hours, exercising for a target duration, or making progress on a work project. The advantage of this restriction is that it makes it possible to derive closed-form characterizations of the agent's behavior and efficient solutions for intervention design problems such as optimal goal setting and reward scheduling. This tractability is important because it allows us to obtain explicit design rules rather than relying solely on numerical simulation or solving a computationally hard planning problem.

Existing tractable models of this kind, however, have two important limitations. First, they are formulated in discrete time. This requires the modeler to choose a time step, such as a day, a week, or a month. If the time step is too coarse, the model cannot represent fine-grained changes in motivation or effort. If the time step is too small, the state space and computational burden grow. Second, existing models typically focus on quasi-hyperbolic discounting or its special cases. Quasi-hyperbolic discounting is mathematically convenient in discrete time, but it represents only one particular form of present bias \cite{phelps1968second,laibson1997golden}. Human discounting behavior is often modeled by richer classes of discount functions, including hyperbolic discounting \cite{ainslie1975specious} and generalized hyperbolic discounting \cite{loewenstein1992anomalies,al2006note}. These limitations motivate the development of a continuous-time model that remains analytically tractable while accommodating a broader class of discount functions.

This paper develops such a model. Our first contribution is to propose a tractable continuous-time model for progress-based tasks performed by time-inconsistent agents.
We formulate the agent's behavior through a variational principle \cite{berdichevsky2009variational}. At each time, the agent chooses a future progress trajectory that minimizes perceived cost, taking into account both discounted effort cost and delayed reward. The agent then follows the infinitesimal initial direction of this perceived-cost-minimizing trajectory and repeats the same procedure continuously over time. This formulation provides a natural continuous-time analogue of existing time-inconsistent planning models, in which an agent repeatedly reoptimizes according to time-discounted preferences. At the same time, the continuous-time setting introduces a nontrivial mathematical challenge: the agent's realized behavior is defined through an infinite sequence of infinitesimal reoptimizations over future trajectories.

Our second contribution is a characterization of the agent's behavior under generalized hyperbolic discounting. Despite the apparent complexity of the continuous-time dynamics, we show that the agent's realized trajectory admits a concise analytical representation. Generalized hyperbolic discounting is a broad class of discount functions that includes exponential discounting as a limiting case and hyperbolic discounting as a special case \cite{loewenstein1992anomalies}. Under this class, we reduce the analysis of abandonment to the shape of a single function. This yields a necessary and sufficient characterization of three qualitatively different regimes: immediate abandonment, successful completion, and abandonment after partial progress. The last regime is precisely the time-inconsistent case in our model: the agent begins the task because completion appears worthwhile from the initial viewpoint, but later abandons it after reoptimizing from a later viewpoint. This characterization is particularly important for intervention design, because it identifies cases in which the agent initially finds completion worthwhile but later needs support to follow through on that plan.

Our third contribution is the analysis of optimal goal setting. Given a fixed time horizon and reward, the intervention designer chooses a goal that maximizes the agent's final progress. We study both non-exploitative and exploitative settings. In the non-exploitative setting, the goal must be attainable and the agent must receive the reward upon completion. In the exploitative setting, the designer is allowed to set an unattainable goal that induces progress even though the agent ultimately fails to obtain the reward. Such interventions are ethically problematic because they exploit the agent's time inconsistency, but analyzing them is important for understanding when and how institutional designs can take advantage of such behavior. We derive optimal goals in both settings and provide conditions under which exploitative rewards are ineffective.

Our fourth contribution is the analysis of optimal reward scheduling. Given a total time horizon and a total reward budget, the designer may divide the task into multiple stages, each with its own subgoal and reward. We show that, when the number of stages is fixed, the optimal schedule divides both time and reward equally across stages. We further prove that increasing the number of stages monotonically improves the agent's final progress, and that the limiting attainable progress is independent of the parameters of the discount function. This result suggests that sufficiently fine reward splitting can mitigate the effect of time discounting in continuous time. It also contrasts sharply with existing discrete-time models with quasi-hyperbolic discounting, where the optimal reward schedule can depend qualitatively on the discount parameters and may sometimes favor not subdividing the task \cite{akagi2024analytically,akagi2026delta}.

A preliminary version of this paper appeared in the Proceedings of AAAI'25 \cite{akagi2025continuous}. The present paper substantially extends that work in two directions. First, it broadens the class of discount functions from exponential and hyperbolic discounting to generalized hyperbolic discounting, which includes both as special cases. Second, it removes the restriction to the special cost parameter $\alpha=2$ and develops the analysis for arbitrary $\alpha>1$. These extensions lead to substantially more general characterizations of time-inconsistent behavior and optimal interventions. They are also technically nontrivial, requiring a careful analysis of the shape and convergence of functions defined through integrals.

The remainder of the paper is organized as follows. Section~2 reviews related work. Section~3 defines the continuous-time progress-based task model. Section~4 characterizes the agent's trajectory and abandonment behavior under generalized hyperbolic discounting. Section~5 studies the optimal goal-setting problem. Section~6 studies the optimal reward-scheduling problem. Section~7 concludes, and Section~8 provides proofs of the mathematical results.

\section{Related Work}
\paragraph{Time-inconsistent planning and intervention design.}
Our work is most closely related to computational models of time-inconsistent planning. \citet{kleinberg2014time} introduced a graph-based model in which a present-biased agent repeatedly chooses a path of minimum perceived cost. This model captures characteristic behaviors such as procrastination and task abandonment, and has led to a line of work at the interface of economics and computation, including planning problems for sophisticated agents, multiple biases, variable present bias, and stochastic behavioral biases \cite{kleinberg2016planning,kleinberg2017planning,gravin2016procrastination,kleinberg2021stochastic}. A key feature of this line of work is that time-inconsistent behavior is treated not only as a descriptive phenomenon, but also as a target of algorithmic intervention. For example, prior studies have examined motivating subgraphs, penalties, and other interventions that change the task environment faced by the agent \cite{albers2019motivating,albers2021value,tang2017computational}. These studies show that intervention design for time-inconsistent agents raises rich computational questions, but they also show that many natural problems in general graph-based models are computationally intractable.

\paragraph{Progress-based tasks.}
To obtain stronger tractability, \citet{akagi2024analytically} introduced a progress-based model for present-biased agents. In this model, an agent accumulates a real-valued amount of progress over a finite horizon and receives a reward if the progress reaches a prescribed goal. This restriction does not cover all planning problems, but it captures an important class of cumulative-effort tasks and enables closed-form characterizations of behavior and efficient solutions for intervention design problems. The model is also closely related to economic models of procrastination, incentives, and self-control in long-term projects \cite{o1999doing,o2006incentives,o2008procrastination}. 
Subsequent work extended this progress-based approach to $\beta$--$\delta$ discounting, a standard discrete-time model of present bias \cite{laibson1997golden}, and showed that the long-run discount factor $\delta$ can qualitatively change optimal interventions \cite{akagi2026delta}.
The present paper continues this line of work, but moves from discrete time to continuous time and from restricted discounting models to generalized hyperbolic discounting. This allows us to analyze intervention design without requiring an arbitrary time discretization, while still preserving analytical tractability.

\paragraph{Discount functions and time-inconsistent preferences.}
Time discounting and time-inconsistent preferences have a long history in economics and behavioral economics \cite{strotz1955myopia,frederick2002time,camerer2004behavioral,wilkinson2017introduction}. Exponential discounting has traditionally served as the standard model of intertemporal choice \cite{samuelson1937note}, but it cannot explain time-inconsistent behavior. This limitation motivated alternative discounting models, including hyperbolic discounting \cite{ainslie1975specious}, quasi-hyperbolic discounting \cite{phelps1968second,laibson1997golden}, and generalized hyperbolic discounting \cite{loewenstein1992anomalies}. Quasi-hyperbolic discounting is particularly convenient in discrete-time models and has been widely used to model present bias. Generalized hyperbolic discounting, by contrast, provides a continuous-time class that includes exponential discounting as a limiting case and hyperbolic discounting as a special case. It has also been characterized axiomatically through impatience and the common difference effect with quadratic delay \cite{al2006note}. In this paper, generalized hyperbolic discounting plays a central role because it is broad enough to include classical discounting models, yet structured enough to permit a detailed analysis of abandonment and intervention design.

\paragraph{Continuous-time models of time inconsistency.}
A separate line of work studies time-inconsistent decision making in continuous time, often through formulations related to optimal control. Since non-exponential discounting typically makes globally optimal plans dynamically inconsistent, these studies often characterize behavior using equilibrium concepts. For example, Ekeland and Lazrak \cite{ekeland2006being,ekeland2008equilibrium} studied subgame-perfect equilibria in continuous time with non-constant discounting, and \citet{karp2007nonconstant} derived a dynamic programming equation for differentiable Markov-perfect equilibria. This line of research has been further developed in general continuous-time equilibrium-control frameworks and refinements of equilibrium concepts \cite{bjork2017time,he2021equilibrium}. Our work is related to these studies in that it also treats time inconsistency in continuous time, but it differs in two important respects. First, these studies primarily focus on sophisticated agents who anticipate future deviations and choose equilibrium strategies. In contrast, our model is closer to naivete: at each time, the agent chooses a perceived-cost-minimizing future trajectory without anticipating that future selves will later reoptimize. Second, our main objective is not to characterize equilibrium controls in general dynamic decision problems, but to obtain tractable behavioral characterizations that can be used for goal and reward design in progress-based tasks.

\paragraph{Connections to reinforcement learning.}
Our framework is also related to reinforcement learning, where agents act so as to maximize discounted rewards \cite{kaelbling1996reinforcement}. The aims, however, are different. Reinforcement learning primarily studies how an agent learns effective behavior in an unknown environment, whereas our goal is to characterize the behavior of a time-inconsistent agent under a specified behavioral model and to design interventions for that agent. These perspectives are complementary. Recent work has also studied reinforcement learning with non-exponential discounting and related identifiability questions \cite{fedus2019hyperbolic,schultheis2022reinforcement,skalse2025partial}, suggesting possible future connections between learning-based models and the intervention-design approach developed here.

\section{Model Definition} \label{sec:model definition}
\subsection{Progress-based Task} \label{subsec:progress-based task}
Following the progress-based task framework \cite{akagi2024analytically}, we consider a task in which an agent aims to reach a goal within a time limit by accumulating a real-valued quantity called progress.
We assume that progress never decreases. 
Such tasks frequently occur in daily life.
For example, consider a student who aims to complete an assignment within 8 hours on a given day. If the student finishes the assignment, they will earn school credit for the course. In this case, the period corresponds to the 8 hours, progress corresponds to the percentage of the assignment completed, and the reward corresponds to the value of earning the course credit.
Another example could be a person participating in a health program that involves exercising for 30 hours within a month. In this scenario, the period is one month, and progress is measured by the hours spent exercising. The reward is the incentive provided upon completing the health program.

The agent is rewarded with $R$ if it achieves the goal progress $\theta$ within the period $T$, where $R, \theta, T \in \R_{> 0}$. 
We denote the agent's state by a tuple $(t, x)$, where $t \in [0, T]$ is the time and $x \in \R_{\geq 0}$ is the progress. 
Note that time $t$ is not restricted to discrete values but takes continuous values, a significant difference from the previous study \cite{akagi2024analytically}.
\cref{fig:task_decription} describes an illustrative example of a progress-based task.

\begin{figure}[t]
  \centering
  \includegraphics[width=0.6\linewidth]{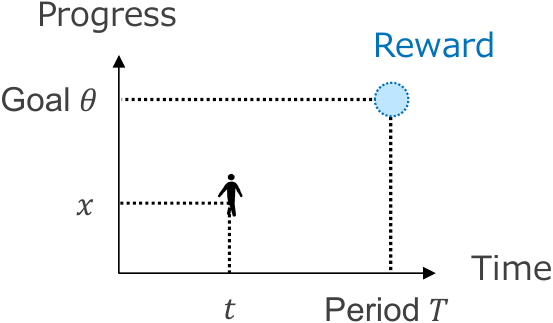}%
  \caption{An illustrative example of a progress-based task.} 
  \label{fig:task_decription}
\end{figure}

\subsection{Perceived Cost} \label{subsec:perceived cost}
We will denote by $\mathcal{Y}_{t, x}$ the set of absolutely continuous functions defined in $[t, T]$ which satisfies $y(t) = x$. 
A function $y(s) \in \mathcal{Y}_{t, x}$ represents the progress trajectory that an agent in state $(t,x)$ will take at future time $s \in [t, T]$.
For progress trajectory $y(s)$, we define \emph{perceived cost} $\mathcal{P}[y]$ by 
\begin{align} \label{eq:perceived cost}
 &\mathcal{P}[y] \coloneqq \int_{t}^T L_t(s, \diff{y}(s)) ds - \dgamma{t}{T} R \cdot \bm{1}[y(T) \geq \theta], \label{eq:perceived cost} \\
  &L_t(s, v) \coloneqq \dgamma{t}{s} c\prn*{v} \label{eq:L}, 
\end{align}
where $\diff{y}(s)$ denotes the derivative function of $y(s)$. 
Note that $\mathcal{P}[\cdot]$ is a \emph{functional} that takes a function as an argument and returns a real value.
The function $\gamma_t: \mathbb{R} \to (0, 1]$  is the \emph{discount function}, which represents how the agent discounts the future benefit or cost; an agent at time $t$ calculates the benefits and costs at a future time $s$ by multiplying the discount factor $\gamma_t(s)$.
The function $c:\R \to \R \cup \set*{+\infty}$ is the \emph{cost function}; to increase progress at a rate of $v$, a cost of $c(v)$ is incurred. 
Also, $\bm{1}[\cdot]$ is an indicator function; $\bm{1}[y(T) \geq \theta] = 1$ if $y(T) \geq \theta$, otherwise $\bm{1}[y(T) \geq \theta] = 0$.

The intuitive interpretation of perceived cost \cref{eq:perceived cost} is the total cost of a future progress trajectory to an agent, considering the discount function effect. 
The first term of \cref{eq:perceived cost} is the total cost perceived by the agent over the period from $t$ to $T$, considering time discounting. It is calculated by integrating the agent's perceived cost $L_t(s, \diff{y}(s))$ at time $s$ from $t$ to $T$. 
The second is the reward term; the agent receives the time-discounted reward $\gamma_t(T)R$ if and only if the final progress $y(T)$ exceeds the goal progress $\theta$. 
This term has a negative sign because rewards have the opposite effect of costs. 

\subsection{Agent Behavior Model} \label{subsec:agent behavior model}
The agent at state $(t, x)$ selects the trajectory $y^*_{t, x}$ with the minimum perceived cost among the candidate trajectories $\mathcal{Y}_{t, x}$. 
If the cost-minimizing trajectory is not unique, we choose one that maximizes 
$\left. \frac{d y^*_{t, x}(s)}{d s} \right|_{s=t}$.
Then, at time $t$, the agent generates progress along $y^*_{t, x}$; it produces progress of 
$\prn*{\left. \frac{d y^*_{t, x}(s)}{d s} \right|_{s=t}} \cdot dt$
during the infinitesimal time $dt$. Starting from the state $(0, 0)$, the agent repeats this procedure at each time $t \in [0, T]$.
\cref{fig:agent_behavior} shows an illustrative example of this process.

\begin{figure}[t]
  \centering
  \includegraphics[width=0.9\linewidth]{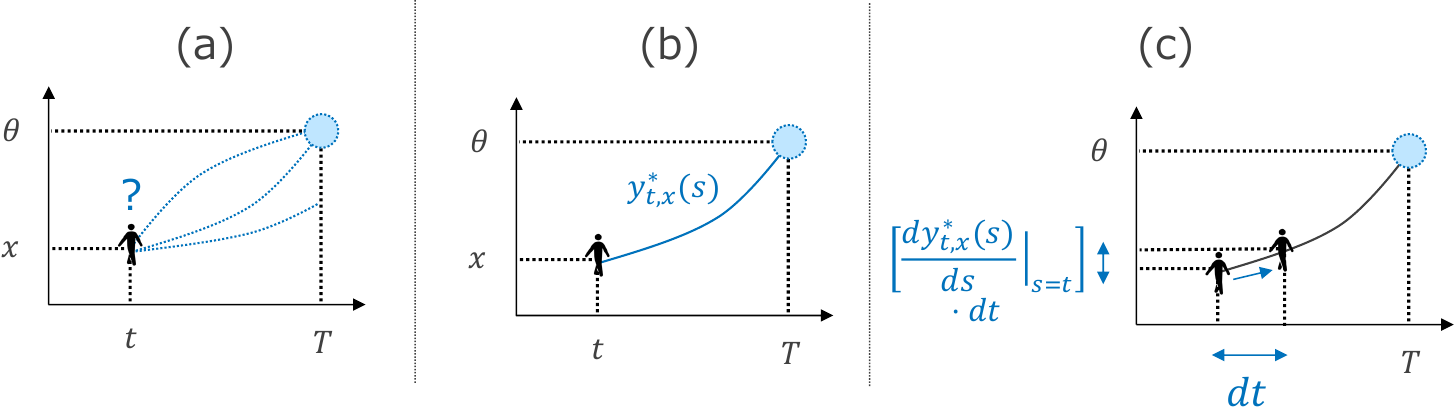}\hfill%
  \caption{
    An illustrative example of the agent behavior model. 
    (a) The agent searches for a trajectory that minimizes the perceived cost among the candidate trajectories of future states $\mathcal{Y}_{t,x}$. The blue dashed line is an example of a candidate trajectory in $\mathcal{Y}_{t,x}$.
    (b) Let $y^*_{t,x}(s)$ denote the trajectory found to minimize the perceived cost.
    (c) The agent generates progress along $y^*_{t,x}(s)$. After this infinitesimal update, the agent returns to the situation in (a) and re-solves the same optimization problem from the new state. 
    Note that while the time interval $dt$ is inherently infinitesimal, it is depicted as larger for illustrative purposes.
  } 
  \label{fig:agent_behavior}
\end{figure}

In this agent behavior model, the trajectory $x(t)$ actually taken by the agent is fully determined by \cref{eq:perceived cost,eq:L} and 
\begin{align}
  y^*_{t, x(t)} &\coloneqq \argmin_{y \in \mathcal{Y}_{t, x(t)}} \mathcal{P}[y], \label{eq:y^*} \\
  \frac{dx(t)}{dt} &= \left. \frac{d y^*_{t, x(t)}(s)}{ds} \right|_{s = t}, x(0) = 0. \label{eq:agent_ODE}
\end{align}

This agent's behavior model is a natural extension of the discrete-time behavior models from previous research \cite{kleinberg2014time,akagi2024analytically} to continuous time. In the previous models, the agent repeatedly chooses the path with the minimum perceived cost among all possible future paths at each time step and acts along that path.
In the model proposed in this paper, the path with the smallest perceived cost corresponds to $y^*_{t, x(t)}$ in \cref{eq:y^*} and acting along the path corresponds to \cref{eq:agent_ODE}. 

The formulation of the agent's behavior in \cref{eq:y^*} is based on the variational principle. 
The variational principle states that a system can be characterized by a function that minimizes or maximizes a given functional \cite{berdichevsky2009variational, kim2021fast, zappala2023neural}.
In our formulation, the agent's future trajectory $y_{t, x(t)}^*$ is defined as the value that minimizes the functional $\mathcal{P}$, making it a clear example of the variational principle. By employing this formulation, we achieve a natural extension of existing agent behavior models to continuous time while also simplifying the analysis of agent behavior by utilizing the well-established variational methods \cite{rockafellar2001convex}.

\subsection{Cost Function} \label{subsec:cost function}
This paper assumes that the cost function can be written as 
\begin{align}
 \label{eq:cost function}
 c(\Delta)
 = 
 \begin{dcases*}
 \Delta^\alpha
 & if $\Delta \geq 0$,\\
 +\infty
 & otherwise, 
 \end{dcases*}
\end{align}
where $\alpha > 1$ is a real parameter. 
Because this paper considers only tasks in which progress never decreases, we set $c(\Delta) = +\infty$ when $\Delta < 0$. 
We can control the shape of the cost function by adjusting the parameter $\alpha$.
This cost function is precisely the same as the one used in the previous study \cite{akagi2024analytically}.

\subsection{Discount Function} \label{subsec:discount function}
The choice of a discount function has a substantial impact on an agent’s behavior. 
We impose the following natural assumptions on $\gamma_t$.
\begin{enumerate}
    \item $\gamma_t(s)$ is continuously differentiable on $[t, T]$.
    \item $(t, s) \mapsto \gamma_t(s)$ is continuous on $\set*{(t, s) \mid 0 \leq t \leq s \leq T}$.
    \item $\gamma_t(s)$ is non-increasing on $[t, T]$.
    \item $\gamma_t(t) = 1$.
\end{enumerate}
This paper mainly focuses on \emph{generalized hyperbolic discounting} \citep{loewenstein1992anomalies}, which is given by
\begin{align} \label{eq:ghyp}
    \gamma_t^{\mathrm{ghyp}}(s) \coloneqq \prn*{1 + \lambda (s - t)}^{-\frac{\mu}{\lambda}} .
\end{align}
Here, $\lambda, \mu \in \R_{>0}$ are parameters.

Generalized hyperbolic discounting constitutes a broad class of discount functions that subsumes historically important and still widely used forms such as exponential \citep{samuelson1937note} and hyperbolic \citep{ainslie1975specious} discounting. In particular, generalized hyperbolic discounting converges to exponential discounting in the limit $\lambda \to +0$, and it coincides with hyperbolic discounting when $\lambda = \mu$. Moreover, it is known that the class of discount functions that simultaneously satisfy two properties often observed in human discounting, namely \emph{impatience} and the \emph{common difference effect with quadratic delay}, coincides with generalized hyperbolic discounting \citep{al2006note}. \cref{fig:ghyp} shows the shapes of the generalized hyperbolic discount functions. 

\begin{figure}[t]
  \centering
  \begin{minipage}{0.48\linewidth}
    \centering
    \includegraphics[width=\linewidth]{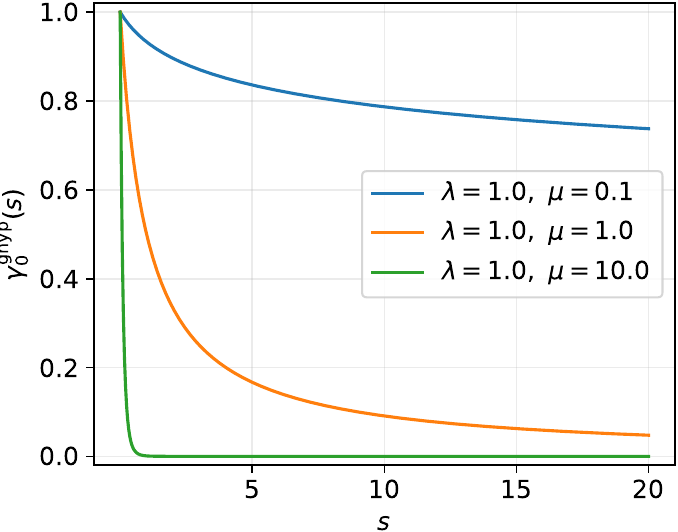}
  \end{minipage}\hfill
  \begin{minipage}{0.48\linewidth}
    \centering
    \includegraphics[width=\linewidth]{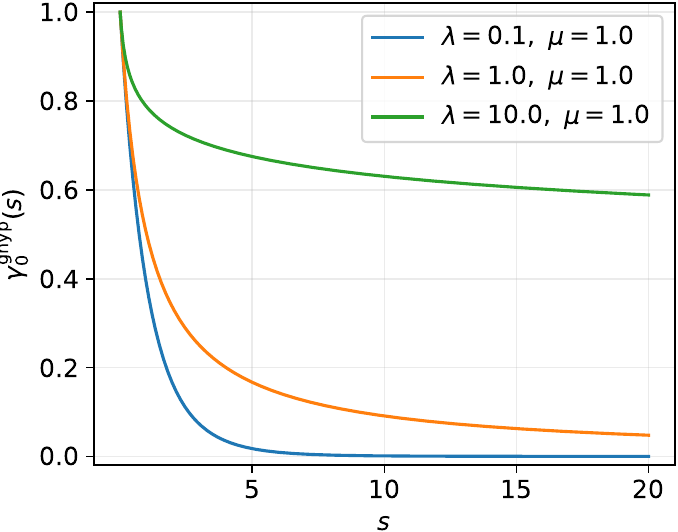}
  \end{minipage}
  \caption{Plots of the generalized hyperbolic discount function $\gamma_0^{\mathrm{ghyp}}(s)$. In the left panel, $\lambda$ is fixed to 1.0 and $\mu = 0.1, 1.0, 10.0$. In the right panel, $\mu$ is fixed to 1.0 and $\lambda = 0.1, 1.0, 10.0$.} \label{fig:ghyp}
\end{figure}

\section{Properties of the Model} \label{sec:properties}
\subsection{Closed Formula of Agent's Trajectory}
The trajectory of the agent $x(t)$ is determined by \cref{eq:perceived cost,eq:L,eq:y^*,eq:agent_ODE}. 
The resulting dynamics involve both a variational problem and a differential equation, and may therefore appear difficult to analyze. Nevertheless, we show that, under mild conditions, the agent's trajectory $x(t)$ can be expressed concisely.
For simplicity, we introduce the following notations; 
\begin{align}
  \Gamma_{t}(s) &\coloneqq \int \dgamma{t}{s}^{-\frac{1}{\alpha - 1}} ds, \label{eq:Gamma}\\
  \zeta(t) &\coloneqq \dgamma{t}{T} \prn*{\dGamma{t}{T} - \dGamma{t}{t}}^{\alpha - 1} \label{eq:zeta}. 
\end{align}
Note that \cref{eq:Gamma} defines $\Gamma_t(s)$ only up to an additive constant. Since the subsequent analysis depends only on differences of the form
$\Gamma_t(s_1)-\Gamma_t(s_2)$, this ambiguity does not affect any of the results.


\begin{theorem}\label{thm:trajectory}
Assume that $\zeta(t)$ is non-increasing on $[0,T]$. Then the agent's trajectory is given by
\begin{align}
x(t)
=
\theta \left(1-\exp\left(-\int_0^{\min\{t,t^*\}}
\frac{ds}{\Gamma_s(T)-\Gamma_s(s)}\right)\right),
\label{eq:analytical_trajectory}
\end{align}
where $t^*$ is defined as follows: If there exists $t \in [0,T]$ such that
\begin{align}
\frac{
\exp\left(-\alpha \int_0^t \frac{ds}{\Gamma_s(T)-\Gamma_s(s)}\right)
}{
\zeta(t)
}
>
\frac{R}{\theta^\alpha},
\label{eq:abandonment_condition}
\end{align}
then $t^*$ is the infimum of all such $t$; otherwise, $t^*=T$.
\end{theorem}

\cref{thm:trajectory} shows that we can calculate the agent's trajectory if the discount function $\gamma_t$ satisfies the property that $\Gamma_t(s)$ and $\int_0^{\min (t, t^*)} \frac{ds}{\Gamma_{s}(T) - \Gamma_{s}(s)}$ are analytically or numerically computable. 
The value $t^*$ is the \emph{abandonment time} when the agent gives up making further progress because any additional progress would not be worth the reward. 
The condition that $\zeta(t)$ is non-increasing corresponds to the scenario where once an agent abandons a task, they do not resume it (see \cref{subsec:proof of thm trajectory}). If this condition does not hold, the agent would repeatedly abandon and resume the task, resulting in more complex behavior. While this case is also interesting, we leave a detailed analysis of it for future work, and this paper focuses on the case where $\zeta(t)$ is non-increasing. 
Note that $x(0)=0$, and under the assumption on $\gamma_t(s)$ stated in \cref{subsec:discount function}, we can show that $x(T)=\theta$ if $t^*=T$.

We explain the derivation of \cref{thm:trajectory} briefly. First, we use the variational method to solve for \cref{eq:y^*}. Typically, the variational method can only find local optima. However, in this case, due to the special structure of the cost function \cref{eq:cost function}, the functional $\mathcal{P}$ becomes a convex functional. This property allows us to obtain the global optimum  $y^*_{t, x(t)}$ by solving a differential equation called the Euler--Lagrange equation \cite{rockafellar2001convex}. 
We finally derive \cref{eq:analytical_trajectory} by solving the differential equation \cref{eq:agent_ODE} using the obtained $y^*_{t, x(t)}$. 
For details, please refer to \cref{subsec:proof of thm trajectory}.

The following lemma verifies that the assumption in \cref{thm:trajectory} is satisfied under generalized hyperbolic discounting.
\begin{lemma} \label{thm:zeta non-increasing}
If $\gamma_t = \gamma_t^{\mathrm{ghyp}}$ defined as \cref{eq:ghyp}, then $\zeta(t)$ is non-increasing on $[0, T]$.
\end{lemma}

Combining \cref{thm:trajectory,thm:zeta non-increasing} with explicit calculations yields the following theorem on the agent's trajectory under generalized hyperbolic discounting.

\begin{theorem} \label{thm:trajectory_ghyp}
Define $\eta(t)$ and $f_T(t)$ by
\begin{align}
    \eta(t) &\coloneqq \int_0^{t}   \frac{ds}{\prn*{1 + \lambda(T-s)}^{\frac{\mu}{\lambda(\alpha-1)} + 1} - 1}, \label{eq:def eta} \\
    f_T(t) &\coloneqq \prn*{\frac{\mu + \lambda(\alpha - 1)}{(\alpha-1)\bra*{\prn*{1 + \lambda(T-t)}^{\frac{\mu}{\lambda(\alpha-1)} + 1}-1}}}^{\alpha-1} \cdot \\ 
    &\prn*{1 + \lambda (T - t)}^{\frac{\mu}{\lambda}} \exp \prn*{- \alpha \prn*{\frac{\mu}{\alpha-1} + \lambda} \eta(t)} \label{eq:def f}. 
\end{align}
When $\gamma_t = \gamma_t^{\mathrm{ghyp}}$ is defined as \cref{eq:ghyp}, the agent's trajectory is given by 
\begin{align}
    x(t) = \theta \prn*{1 - \exp \prn*{- \prn*{\frac{\mu}{\alpha-1} + \lambda}  \eta \prn*{\min(t, t^*)}}}, \label{eq:trajectory_ghyp} 
\end{align}
where $t^*$ is defined as follows: If there exists $t \in [0, T]$ such that
\begin{align}
    f_T(t) > \frac{R}{\theta^{\alpha}} \label{eq:def of t*}, 
\end{align}
then $t^*$ is the infimum of all such $t$; otherwise, $t^* = T$. 
\end{theorem}

In \cref{thm:trajectory_ghyp}, the main obstacle to an analytic treatment is the function $\eta(t)$. The function $\eta(t)$ belongs to the class of Gauss hypergeometric functions \citep{abramowitz1964handbook}, and the integral cannot be expressed by elementary functions except in special cases, for example when $\frac{\mu}{\lambda(\alpha-1)} = 1$.
However, as we show later, even though such a complicated function is involved, we can still analyze the agent's behavior mathematically and discuss optimal intervention strategies.
Note that $x(0) = 0$, and $\lim_{t \to T}x(t) = \theta$ if $t^* = T$. 

In what follows, we restrict attention to the generalized hyperbolic discounting case, that is, $\gamma_t = \gamma_t^{\mathrm{ghyp}}$.

\subsection{Abandonment Time and Time Inconsistency}
In this section, we investigate the abandonment time $t^*$ in \cref{thm:trajectory_ghyp}, which represents the time at which the agent gives up. The value of $t^*$ is important not only because it determines the trajectory that the agent follows during the task, but also because it determines whether the agent exhibits time-inconsistent behavior.
Indeed, when $t^* = 0$, the agent gives up from the outset and does not start the task, whereas when $t^* = T$, the agent completes the task. In either case, the agent behaves time-consistently. In contrast, when $0 < t^* < T$, the agent starts the task expecting to complete it and receive the reward, but later abandons it after judging that the cost of completion is no longer worth the reward. Thus, the agent exhibits time-inconsistent behavior.

Since $t^*$ is defined as the infimum of all $t$ satisfying \cref{eq:def of t*}, with $t^* = T$ if no such $t$ exists, understanding the behavior of $t^*$ requires analyzing the shape of $f_T(t)$ on $[0, T]$. The following lemma characterizes the shape of $f_T(t)$, in particular where its maximum is attained.
\begin{lemma} \label{thm:shape of f_T(t)}
    The equation
    \begin{align}
        (\alpha - 1) x^{\frac{\mu}{\lambda} + \alpha - 1}
        - \alpha \prn*{\frac{\mu}{\lambda(\alpha - 1)} + 1}x^{\alpha - 1}
        + \frac{\mu}{\lambda} = 0 \label{eq:x_0 equation}
    \end{align}
    has a unique real solution with $x > 1$. Let $x_{\sol}$ denote this solution, and define
    \begin{align}
        t_{\sol} \coloneqq T - \frac{x_{\sol}^{\alpha - 1}-1}{\lambda}.
    \end{align}
    \begin{enumerate}
        \item If $t_{\sol} \leq 0$, then $f_T(t)$ is strictly decreasing on $[0, T]$. In particular, its maximum is attained at $t = 0$.
        \item If $t_{\sol} > 0$, then $f_T(t)$ is strictly increasing on $[0, t_{\sol}]$ and strictly decreasing on $[t_{\sol}, T]$. In particular, its maximum is attained at $t = t_{\sol}$.
    \end{enumerate}
\end{lemma}

The following theorem follows immediately from \cref{thm:shape of f_T(t)} and the definition of $t^*$.



\begin{theorem} \label{thm:value of t^*}
    \begin{enumerate}
        \item If $t_{\sol} \leq 0$, then
        \begin{align}
            t^* =
            \begin{cases}
                0 & \text{if } f_T(0) > \frac{R}{\theta^{\alpha}}, \\
                T & \text{otherwise}.
            \end{cases} \label{eq:t star tsol small}
        \end{align}

        \item If $t_{\sol} > 0$, then
        \begin{align}
            t^* =
            \begin{cases}
                0 & \text{if } f_T(0) > \frac{R}{\theta^{\alpha}}, \\
                t' & \text{if } f_T(0) \leq \frac{R}{\theta^{\alpha}} < f_T(t_{\sol}), \\
                T & \text{otherwise},
            \end{cases} \label{eq:t star tsol large}
        \end{align}
        where $t' \in [0, t_{\sol}]$ is the unique value satisfying
        \begin{align}
            f_T(t') = \frac{R}{\theta^{\alpha}}. \label{eq:t' def}
        \end{align}
    \end{enumerate}
\end{theorem}

\Cref{thm:value of t^*} implies the following. If $t_{\sol} \leq 0$, that is, if $T \leq \frac{x_{\sol}^{\alpha - 1}-1}{\lambda}$, then the agent behaves time-consistently regardless of the values of the goal $\theta$ and the reward $R$. In contrast, if $t_{\sol} > 0$, that is, if $T > \frac{x_{\sol}^{\alpha - 1}-1}{\lambda}$, then the agent may exhibit time-inconsistent behavior depending on the values of $\theta$ and $R$. Since $\frac{x_{\sol}^{\alpha - 1}-1}{\lambda}$ does not depend on $T$, time-inconsistent behavior is more likely to arise when $T$ is large.

Furthermore, we can analyze the agent's behavior as $\lambda \to 0$, that is, as the discount function approaches exponential discounting. To explain this, we first present the following lemma.

\begin{lemma} \label{thm:x convergence}
    $\lim_{\lambda \to 0} \frac{x_{\sol}^{\alpha - 1} - 1}{\lambda} = \infty$. 
\end{lemma}

\Cref{thm:x convergence} indicates that we have
$
    T \leq \frac{x_{\sol}^{\alpha - 1}-1}{\lambda}
$
for sufficiently small $\lambda$. 
This implies that the agent behaves time-consistently as the discount function approaches exponential discounting.
This result is consistent with the well-known fact that exponential discounting does not induce time-inconsistent behavior \cite{frederick2002time}.

\section{Optimal Goal-setting Problem} \label{sec:goal setting}
In this section, we consider the problem of finding the goal $\theta^*$ that maximizes the final progress $x(T)$:
\begin{align}
  \theta^* \coloneqq \argmax_{\theta \geq 0} x(T),
  \label{eq:problem_goal_setting}
\end{align}
given the period $T$, reward $R$, and parameters $\alpha$, $\lambda$, and $\mu$. This problem captures how an intervener should choose a goal so as to maximize an agent's progress over a fixed time horizon under a given reward. We refer to this as the \emph{optimal goal-setting problem}.

This problem is of practical importance in many real-world situations. For example, consider a company president, acting as the intervener, who offers employees a fixed bonus for achieving a sales goal within six months. A natural question is how this goal should be set to maximize final sales. If the goal is set too high, employees may lose motivation because the goal appears unattainable, leading to lower final sales. By contrast, if the goal is set too low, employees may reach it too easily, obtain the bonus, and reduce their effort thereafter, again resulting in lower sales. These considerations illustrate the importance of choosing an appropriate goal, as formulated in \cref{eq:problem_goal_setting}.

The solution to the goal-setting problem depends crucially on whether \emph{exploitative rewards} are allowed. Exploitative rewards are decoy rewards that take advantage of an agent's time inconsistency. Because a time-inconsistent agent may make substantial progress even when ultimately abandoning the goal, setting a high and unattainable goal can in some cases induce greater progress than setting a lower and achievable one, despite the agent's eventual failure to complete the task. Such a policy may also benefit the intervener, since no reward needs to be paid when the agent does not achieve the goal. At the same time, exploitative rewards raise ethical concerns, for example by reducing motivation or undermining long-term effort, and therefore require careful consideration. The significance of exploitative rewards has been discussed in the literature \cite{kleinberg2014time,tang2017computational,albers2019motivating}. In particular, \citet{akagi2024analytically} showed that, for agents with quasi-hyperbolic discounting, exploitative rewards can substantially increase overall progress.

We therefore study both settings, namely, one in which exploitative rewards are allowed and one in which they are disallowed. When exploitative rewards are allowed, the optimization problem is exactly \cref{eq:problem_goal_setting}. When they are disallowed, we impose the additional constraint $x(T)=\theta$, which guarantees that the agent achieves the goal and receives the reward.

\subsection{Non-Exploitative Reward Case}

\begin{theorem} \label{thm:optimal goal non-exploitatie}
When exploitative rewards are not permitted, 
\begin{align}
    \theta^* = 
    \begin{cases}
    R^{\frac{1}{\alpha}} f_T(0)^{-\frac{1}{\alpha}} & t_{\sol} \leq 0, \\
    R^{\frac{1}{\alpha}} f_T(t_{\sol})^{-\frac{1}{\alpha}} & t_{\sol} > 0. 
    \end{cases}
\end{align}
\end{theorem}

\Cref{thm:optimal goal non-exploitatie} implies that, when exploitative rewards are not allowed, the solution to the optimal goal-setting problem can be expressed in a simple form.

Since exploitative rewards are not allowed, the final progress $x(T)$ must equal $\theta$. Therefore, the maximum final progress is equal to $\theta^*$. It follows that the maximum final progress is monotonically increasing in the reward $R$. With respect to $T$, the following lemma holds.
\begin{lemma} \label{thm:f_T decreasing}
$f_T(0)$ and $f_T(t_{\sol})$ are monotonically decreasing in $T$.
\end{lemma}
Therefore, the maximum final progress is also monotonically increasing in the time horizon $T$. These results show that increasing the reward or relaxing the time constraint makes it possible to induce greater progress from the agent.

\subsection{Exploitative Reward Case}

\begin{theorem} \label{thm:optimal goal exploitatie}
We consider the case where exploitative rewards are allowed.
Define $g(t)$ by
\begin{align}
    g(t) \coloneqq f_T(t)^{-\frac{1}{\alpha}} \prn*{1 - \exp \prn*{- \prn*{\frac{\mu}{\alpha-1} + \lambda}  \eta \prn*{t}  }}. 
\end{align}
Since $g$ is continuous on the compact interval $[0,t_{\sol}]$, the maximizer exists. 
Define 
\begin{align}
    \tau^* \coloneqq \argmax_{\tau \in [0, t_{\sol}]} g(\tau). 
\end{align}
Then, the following hold:
\begin{enumerate}
    \item When $t_{\sol} \leq 0$, $\theta^* = R^{\frac{1}{\alpha}} f_T(0)^{-\frac{1}{
\alpha}}$ and $x(T) = \theta^*$. 

    \item When $t_{\sol} > 0$ and $f_T(t_{\sol})^{-\frac{1}{\alpha}} \geq g(\tau^*)$, $\theta^* = R^{\frac{1}{\alpha}} f_T(t_{\sol})^{-\frac{1}{\alpha}}$ and $x(T) = \theta^*$.

    \item  When $t_{\sol} > 0$ and $f_T(t_{\sol})^{-\frac{1}{\alpha}} < g(\tau^*)$, $\theta^* = R^{\frac{1}{\alpha}} f_T(\tau^*)^{-\frac{1}{\alpha}}$ and $x(T) = R^{\frac{1}{\alpha}}g(\tau^*)$. 
\end{enumerate}
\end{theorem}
In cases (a) and (b), $x(T)=\theta^*$ holds, implying that exploitative rewards are ineffective and that non-exploitative rewards are optimal.
In contrast, in case (c), $x(T)\ne \theta^*$, implying that exploitative rewards are optimal: $R^{\frac{1}{\alpha}} g(\tau^*)$, which is the maximum final progress achieved under exploitative rewards, exceeds $R^{\frac{1}{\alpha}} f_T(t_{\sol})^{-\frac{1}{\alpha}}$, which is the maximum final progress achieved under non-exploitative rewards.

We now ask when condition (c) holds. The condition $t_{\sol}>0$ was examined in the previous section, so here we focus on the condition $f_T(t_{\sol})^{-\frac{1}{\alpha}} < g(\tau^*)$.

\begin{theorem} \label{thm:alpha bound}
If $1 < \alpha \leq 1 + \frac{\mu + \sqrt{4 \lambda \mu + \mu^2}}{2 \lambda}$, then $f_T(t_0)^{-\frac{1}{\alpha}} \geq g(\tau^*)$.
\end{theorem}

\Cref{thm:alpha bound} implies that exploitative rewards are ineffective when $\alpha$ satisfies
$
    1 < \alpha \leq 1 + \frac{\mu + \sqrt{4 \lambda \mu + \mu^2}}{2 \lambda}.
$
In other words, exploitative rewards are ineffective when $\alpha$ is sufficiently small.
Since
$
    1 + \frac{\mu + \sqrt{4 \lambda \mu + \mu^2}}{2 \lambda}
    = 1 + \frac{\frac{\mu}{\lambda} + \sqrt{\frac{4\mu}{\lambda} + \prn*{\frac{\mu}{\lambda}}^2}}{2},
$
this upper bound depends only on the ratio $\mu/\lambda$.
As $\lambda \to +0$, generalized hyperbolic discounting approaches exponential discounting, and
$
    \lim_{\lambda \to +0} \prn*{1 + \frac{\mu + \sqrt{4 \lambda \mu + \mu^2}}{2 \lambda}} = +\infty.
$
Therefore, when generalized hyperbolic discounting is close to exponential discounting, exploitative rewards are ineffective regardless of the value of $\alpha$.
When $\lambda = \mu$, generalized hyperbolic discounting coincides with hyperbolic discounting. In this case,
$
    1 + \frac{\mu + \sqrt{4 \lambda \mu + \mu^2}}{2 \lambda}
    = \frac{3 + \sqrt{5}}{2} \approx 2.62,
$
so exploitative rewards are ineffective whenever $1 < \alpha \leq \frac{3 + \sqrt{5}}{2}$.

\Cref{thm:alpha bound} provides a sufficient condition under which exploitative rewards are ineffective. One may then ask whether there exist parameter settings for which exploitative rewards are in fact effective. Such parameter settings do exist. For example, when $\alpha = 3$, $\lambda = 10000$, $\mu = 10000$, and $T = 100$, we obtain
\begin{align}
    t_{\sol} \approx 99.9996, \quad f_T(t_{\sol})^{-\frac{1}{\alpha}} \approx 0.0178, \quad g(95) \approx 0.0234.
\end{align}
Hence, in this case, $t_{\sol} > 0$ and $f_T(t_{\sol})^{-\frac{1}{\alpha}} < g(\tau^*)$. Therefore, exploitative rewards are effective in this case.

\section{Optimal Reward Scheduling Problem} \label{sec:reward scheduling}
This section studies a more advanced intervention design problem, namely, the \emph{optimal reward scheduling problem}. In this setting, the intervener is allowed to split the total reward into multiple parts and distribute them to the agent over time.

The inputs to the problem are the total time horizon $T$, the total reward $R$, and the agent parameters $\alpha$, $\lambda$, and $\mu$. The intervener first chooses a positive integer $N$, and then divides the total time horizon $T$ into $N$ periods $T_1, \ldots, T_N \in \R_{\geq 0}$ and the total reward $R$ into $N$ partial rewards $R_1, \ldots, R_N \in \R_{\geq 0}$. These variables must satisfy
$
    \sum_{i=1}^N T_i = T
$
and
$
    \sum_{i=1}^N R_i = R.
$
The intervener also specifies target progress levels $\theta_i \in \R_{\geq 0}$ for each $i \in \set*{1, \ldots, N}$. We refer to the tuple
$
    \prn*{N, (T_i)_{i=1}^N, (R_i)_{i=1}^N, (\theta_i)_{i=1}^N}
$
as a reward schedule.

Given a reward schedule, the agent accumulates progress over $N$ stages. In stage $i$, the agent attempts to achieve the target progress $\theta_i$ within period $T_i$ in order to obtain reward $R_i$. We assume that progress in each stage is independent of the others. That is, in stage $i$, the agent makes progress without taking into account the rewards or targets associated with any other stage. Let $P_i$ denote the increment of progress in stage $i$. Our objective is to maximize the total progress increment
$
    \sum_{i=1}^{N} P_i
$
by choosing an appropriate reward schedule. 
Here, we consider the setting in which exploitative rewards are \emph{not} allowed. That is, we assume that $P_i = \theta_i$ holds at each stage $i$. 
\Cref{fig:reward schedule} illustrates an example of a reward schedule in this setting.

\begin{figure}[t] 
  \centering 
  \includegraphics[width=0.6\linewidth]{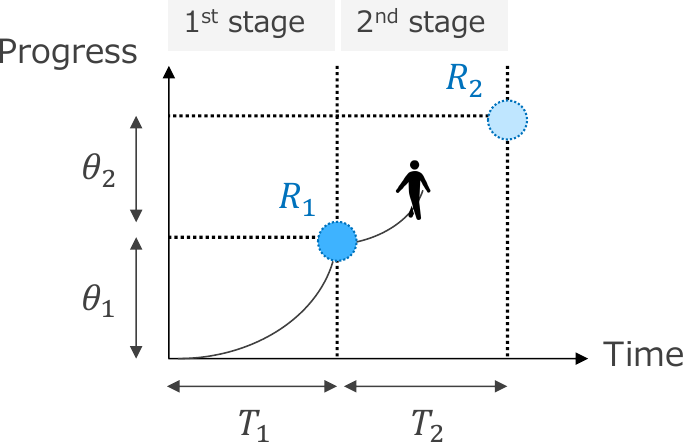}\hfill%
  \caption{An example of reward schedule when $N=2$. } 
  \label{fig:reward schedule}
\end{figure}

\begin{theorem} \label{thm:orsp fixed N}
    When $N$ is fixed, the optimal reward schedule is 
    \begin{align}
        T_i = \frac{T}{N}, \quad R_i = \frac{R}{N}, \quad \theta_i = R_i^{\frac{1}{\alpha}} F(T_i)^{\frac{\alpha-1}{\alpha}}, 
    \end{align}
where 
\begin{align}
    F(T) &\coloneqq 
    \begin{cases}
    F_1(T) & T \leq \frac{x_{\sol}^{\alpha-1}-1}{\lambda}, \\
    F_2(T) & T > \frac{x_{\sol}^{\alpha-1}-1}{\lambda},
    \end{cases} \\
    F_1(T) &\coloneqq f_T(0)^{-\frac{1}{\alpha-1}}, \\
    F_2(T) &\coloneqq f_T\prn*{T - \frac{x_{\sol}^{\alpha-1}-1}{\lambda}}^{-\frac{1}{\alpha-1}}.
    \end{align}
The optimal objective function value $\sum_{i=1}^{N} P_i$ equals to
\begin{align}
G(N) \coloneqq R^{\frac{1}{\alpha}} \prn*{N \cdot F \prn*{\frac{T}{N}}}^{\frac{\alpha - 1}{\alpha}}. 
\end{align}
\end{theorem}

\cref{thm:orsp fixed N} implies that when we fix $N$, we can write the optimal reward schedule in closed form. Furthermore, it is optimal to evenly distribute the total rewards and provide them at regular intervals. 

\begin{theorem} \label{thm:G(N)}
$G(N)$ is monotonically increasing with respect to $N$ and 
\begin{align}
    \lim_{N \to \infty} G(N) = R^{\frac{1}{\alpha}} T^{\frac{\alpha - 1}{\alpha}}. 
\end{align}
\end{theorem}

\cref{thm:G(N)} indicates that finer divisions of rewards and periods lead to greater total progress. However, there is an upper bound to the total progress, and as $N$ increases, the total progress converges to this upper bound.

An important observation is that the limit $R^{\frac{1}{\alpha}} T^{\frac{\alpha - 1}{\alpha}}$ does not depend on the discount function parameters $\lambda$ and $\mu$. 
This implies that providing intermediate rewards at sufficiently fine intervals can counteract the effects of the discount function. 
In practical situations, it is impossible to divide rewards infinitely finely, and there are often constraints on the number of divisions. Even in such cases, splitting the rewards as finely as possible can mitigate the effects of time discounting and maximize the attainable progress.

This result contrasts with the previous findings under quasi-hyperbolic discounting; \cite{akagi2024analytically,akagi2026delta} demonstrated that under discrete time quasi-hyperbolic discounting, varying the reward intervals based on the discount-function parameter is optimal. For specific parameters, providing the rewards in a lump sum (i.e., $N=1$) is optimal, whereas for other parameters, dividing the rewards as finely as possible is optimal. In contrast, in the continuous-time model proposed in this paper, the optimal strategy is to divide the rewards as finely as possible, regardless of the parameters of the discount function. 
This contrast suggests that the choice of time model and discounting model can qualitatively affect the structure of optimal reward schedules. A more detailed comparison between discrete-time and continuous-time formulations is an interesting direction for future work.

\section{Conclusion} \label{sec:conclusion}
This paper developed a tractable continuous-time framework for designing interventions for time-inconsistent agents in progress-based tasks. By formulating the agent's behavior through a variational principle, we obtained a concise analytical representation of the realized trajectory under generalized hyperbolic discounting. This representation enabled us to characterize when the agent completes the task, abandons it immediately, or exhibits time-inconsistent abandonment after making partial progress.

We then used this characterization to study two intervention design problems. For optimal goal setting, we derived optimal goals both with and without exploitative rewards, and identified conditions under which exploitative rewards are ineffective. For optimal reward scheduling, we showed that equal time intervals and equal reward allocations are optimal when the number of stages is fixed, and that finer reward splitting monotonically improves final progress up to a discount-independent limit. These results provide explicit design rules for interventions in continuous-time progress-based tasks and reveal qualitative differences from existing discrete-time models.

Several limitations remain. The model focuses on one-dimensional, nondecreasing progress, which is essential for tractability but excludes more general sequential decision problems with branching states, uncertainty, or decreasing progress. Empirical validation is also an important direction for future work. Testing whether the predicted abandonment regimes and optimal intervention structures arise in human behavior would clarify the practical scope of the proposed framework.

\section{Proofs} \label{sec:proofs}
\subsection{Proof of \cref{thm:trajectory}} \label{subsec:proof of thm trajectory}
\begin{proof}
First, we review an existing result from the literature on the calculus of variations \citep{rockafellar2001convex}.
Let us denote by $\mathcal{A}^p[a, b]$ the set of functions that are absolutely continuous on the interval $[a, b]$ and whose derivatives are $p$-integrable, where $p \in [1, \infty]$. 
We consider the following variational problem; 
\begin{align}
    \label{eq:variational problem strict}
    \begin{aligned}
    \underset{y \in \mathcal{A}^p}{\mathrm{min.}}& \quad \int_{a}^{b} f(s, y, \diff{y}(x)) ds \\
    \mathrm{s.t.}& \quad y(a) = A, \ y(b) = B, 
    \end{aligned}
\end{align}
where $f(x_0, x_1, x_2)$ is continuously differentiable. 
Write 
\begin{align}
f_i(x_0, x_1, x_2) = \frac{\partial}{\partial x_i} f(x_0, x_1, x_2)\ (i=0, 1, 2).
\end{align}
Then, the following theorem holds. 

\begin{theorem}[\citep{rockafellar2001convex}, Theorem 3] \label{thm:thm:variational_optimality_strict}
If the function $f(x_0, x_1, x_2)$ is convex with respect to $(x_1, x_2)$ and there exist functions $y \in \mathcal{A}^p[a, b]$ and $z \in \mathcal{A}^1[a, b]$ that satisfy
\begin{align}
    &z(s) = f_2(s, y, \diff{y}),\ z'(s) =f_1(s, y, \diff{y}), \label{eq:Euler-Lagrange_strict} \\
    &y(a) = A,\ y(b) = B, 
\end{align}
then $y$ is an optimal solution of \cref{eq:variational problem strict}. 
\end{theorem}
The equations \cref{eq:Euler-Lagrange_strict} are called Euler--Lagrange equation. 

Next, we move to the proof of \cref{thm:trajectory}. 
A function $y$ that minimizes $\mathcal{P}[y]$ satisfies $y(T) = \theta$ or $y(T)=x(t)$. 
This is because $y(T) = \theta$ holds when $y(T) \geq \theta$ and $y(T) = x(t)$ holds otherwise from the definitions of functional $\mathcal{P}[y]$ and cost function $c(\Delta)$. 

(i) When $y(T) = x(t)$, it is clear that the optimal choice is $y(s) = x(t)$ for $s \in [t, T]$, and in this case, $\mathcal{P}[y]=0$. 

(ii) When $y(T) = \theta$, because 
\begin{align}
    \mathcal{P}[y] = \int_{t}^T L_t(s, \diff{y}(s)) ds - \dgamma{t}{T} R 
\end{align}
holds, it is sufficient to solve the problem 
\begin{align}
    \min_{y \in \mathcal{A}^1[t, T]} \int_{t}^T L_t(s, \diff{y}(s)) ds \quad \mathrm{s.t.}\  y(t) = x(t),\  y(T) = \theta. \label{eq:variational problem}
\end{align}
Note that $\mathcal{A}^1[t, T]$ consists of all absolutely continuous function in $[t, T]$.  

Here, we consider a new variational problem 
\begin{align}
    \min_{y \in \mathcal{A}^1[t, T]} \int_{t}^T \tilde{L}_t(s, \diff{y}(s)) ds \quad \mathrm{s.t.}\  y(t) = x(t),\  y(T) = \theta, \label{eq:variational problem tilde}
\end{align}
where 
\begin{align}
    \tilde{c}(\Delta) &\coloneqq \abs{\Delta}^{\alpha}, \\
    \tilde{L}_t(s, v) &\coloneqq \gamma_t(s) \tilde{c}(v). 
\end{align}
In this problem, if an optimal solution exists for problem \cref{eq:variational problem tilde}, then an optimal solution also exists for problem \cref{eq:variational problem}, and the optimal solutions coincide. Indeed, any optimal solution to problem \cref{eq:variational problem tilde} must be non-decreasing. Otherwise, replacing the trajectory by a monotone path with the same endpoints would strictly decrease the objective function, 
contradicting optimality. Therefore, solving problem \cref{eq:variational problem tilde} is equivalent to solving problem \cref{eq:variational problem}, and we henceforth consider only problem \cref{eq:variational problem tilde}.
The function $\tilde{L}_t(s, \diff{y}(s)) = \gamma_t(s) \abs{\diff{y}(s)}^{\alpha}$ is continuously differentiable and convex with respect to $(y(s), \diff{y}(s))$. Also, 
\begin{align}
    y(s) &= x(t) + (\theta - x(t)) \cdot \frac{\Gamma_t(s) - \Gamma_t(t)}{\Gamma_t(T) - \Gamma_t(t)} \label{eq:optimal_y}, \\
    z(s) &= \alpha \prn*{\frac{\theta - x(t)}{\Gamma_t(T) - \Gamma_t(t)}}^{\alpha-1}
\end{align}
satisfy the Euler--Lagrange equation \cref{eq:Euler-Lagrange_strict} and $y, z \in \mathcal{A}^1[a, b]$. 
From \cref{thm:thm:variational_optimality_strict}, \cref{eq:variational problem tilde} has an optimal solution, and the solution is given by \cref{eq:optimal_y}. 
Thus, \cref{eq:optimal_y} is also an optimal solution of \cref{eq:variational problem}
, and in this case a simple calculation gives 
\begin{align}
    \mathcal{P}[y] &=  \frac{(\theta - x(t))^{\alpha} - \zeta(t) R }{(\Gamma_t(T) - \Gamma_t(t))^{\alpha-1}}. 
\end{align} 
From (i) and (ii), we have
\begin{align} \label{eq:y^*_solved}
    y^*_{t, x(t)}(s) =
    \begin{cases}
         x(t) & (\theta - x(t))^{\alpha} - \zeta(t) R > 0, \\
        x(t) +  \frac{\prn*{\theta - x(t)} \prn*{\Gamma_t(s) - \Gamma_t(t)}}{\Gamma_t(T) - \Gamma_t(t)} & \mathrm{otherwise}. 
    \end{cases}
\end{align}

Next, we solve \cref{eq:agent_ODE}. By differentiating \cref{eq:y^*_solved}, we get 
\begin{align}
    \frac{d y^*_{t, x(t)}(s)}{ds} =
    \begin{cases}
         0 & (\theta - x(t))^{\alpha} - \zeta(t) R > 0, \\
        (\theta - x(t)) \cdot \frac{\gamma_t(s)^{-\frac{1}{\alpha-1}}}{\Gamma_t(T) - \Gamma_t(t)} & \mathrm{otherwise}.
    \end{cases}
\end{align}
If the condition $(\theta - x(\hat{t}))^{\alpha} - \zeta(\hat{t}) R > 0$ is satisfied for a certain $\hat{t}$, then it follows from the form of \cref{eq:agent_ODE} and non-increasingness of $\zeta(t)$, that $x(t)$ remains constant for $t \in \bra*{\hat{t}, T}$. 
Let $\mt$ denote the infimum of all such $\hat{t}$. 
Then, 
\begin{align}
    \frac{d y^*_{t, x(t)}(s)}{ds} = (\theta - x(t)) \cdot \frac{\gamma_t(s)^{-\frac{1}{\alpha-1}}}{\Gamma_t(T) - \Gamma_t(t)} \label{eq:agent_ODE_simple}
\end{align}
holds for $t \in \bra*{0, \mt}$. 
In this case, solving \cref{eq:agent_ODE} under the condition $x(0) = 0$ yields
\begin{align}
    x(t) = 
    \theta \prn*{1 - \exp \prn*{-\int_0^{t} \frac{ds}{\Gamma_{s}(T) - \Gamma_{s}(s)}}}. 
\end{align}
Because
\begin{align}
    (\theta - x(t))^{\alpha} - \zeta(t) R =
    \theta^{\alpha} \exp \prn*{-\alpha \int_0^{t} \frac{ds}{\Gamma_{s}(T) - \Gamma_{s}(s)}} - \zeta(t) R
\end{align}
holds, $\mt$ is the minimum $t$ that satisfies \cref{eq:abandonment_condition}. 
\end{proof}

\subsection{Proof of \cref{thm:zeta non-increasing}}
\begin{proof}
A simple calculation yields
\begin{align}
    \Gamma_t(s) &= \frac{\alpha-1}{\mu + \lambda(\alpha-1)} \prn*{1 + \lambda(s-t)}^{\frac{\mu}{\lambda(\alpha-1)} + 1}, \\
    \zeta(t) &= \prn*{\frac{\alpha-1}{\mu + \lambda (\alpha-1)}}^{\alpha-1} \prn*{1 + \lambda (T - t)}^{-\frac{\mu}{\lambda}} \bra*{\prn*{1 + \lambda(T-t)}^{\frac{\mu}{\lambda(\alpha-1)} + 1}-1}^{\alpha-1}. 
\end{align}
Since $\frac{\alpha-1}{\mu + \lambda (\alpha-1)} > 0$, it suffices to show that
\begin{align}
    \tilde{\zeta}(x) = x^{-\nu} \prn*{x^{1 + \frac{\nu}{\alpha-1}}-1}^{\alpha-1}
\end{align}
is non-decreasing for $x \ge 1$, where we define $x \coloneqq 1 + \lambda (T - t)$ and $\nu \coloneqq \frac{\mu}{\lambda}$.
Indeed,
\begin{align}
    \frac{d \tilde{\zeta}(x)}{dx} = x^{-\nu-1} \prn*{x^{\frac{\nu}{\alpha - 1}} - 1}^{\alpha - 2} \prn*{(\alpha - 1)x^{\frac{\nu}{\alpha - 1}} + \nu}
\end{align}
is nonnegative for $x \ge 1$ because $\alpha > 1$ and $\nu > 0$, which implies that $\tilde{\zeta}(x)$ is non-decreasing for $x \ge 1$.
\end{proof}

\subsection{Proof of \cref{thm:shape of f_T(t)}}
\begin{proof}
We first compute the derivative of $\log f_T(t)$:
\begin{align}
    \log f_T(t)
    &= (\alpha - 1) \log \prn*{\frac{\mu}{\alpha - 1} + \lambda}
    - (\alpha - 1) \log \bra*{\prn*{1 + \lambda(T-t)}^{\frac{\mu}{\lambda(\alpha-1)} + 1}-1}\\
     &\quad + \frac{\mu}{\lambda} \log (1 + \lambda (T - t))
    - \alpha \prn*{\frac{\mu}{\alpha-1} + \lambda}
    \int_0^{t} \frac{ds}{\prn*{1 + \lambda(T-s)}^{\frac{\mu}{\lambda(\alpha-1)} + 1} - 1},
\end{align}
and hence
\begin{align}
    \frac{d \log f_T(t)}{dt}
    &= \lambda \cdot
    \frac{(\alpha - 1) x^{\frac{\mu}{\lambda} + \alpha - 1}
    - \alpha \prn*{\frac{\mu}{\lambda} \frac{1}{\alpha - 1} + 1}x^{\alpha - 1}
    + \frac{\mu}{\lambda}}
    {x^{\alpha - 1} \prn*{x^{\frac{\mu}{\lambda} + \alpha - 1} - 1}}. 
    \label{eq:diff log f}
\end{align}
where 
$
    x \coloneqq \prn*{1 + \lambda(T-t)}^{\frac{1}{\alpha-1}}.
$
Since the denominator in \cref{eq:diff log f} is positive for $x>1$, the sign of $\frac{d \log f_T(t)}{dt}$ is determined by the sign of 
\begin{align}
    g(x) \coloneqq
    (\alpha - 1) x^{\frac{\mu}{\lambda} + \alpha - 1}
    - \alpha \prn*{\frac{\mu}{\lambda} \frac{1}{\alpha - 1} + 1}x^{\alpha - 1}
    + \frac{\mu}{\lambda}.
\end{align}
Because
\begin{align}
    \frac{dg(x)}{dx}
    = \prn*{\frac{\mu}{\lambda} + \alpha - 1} x^{\alpha - 2}
    \prn*{(\alpha - 1)x^{\frac{\mu}{\lambda}} - \alpha}, 
\end{align}
$g(x)$ is strictly decreasing for
$
    1 \le x < \prn*{\frac{\alpha}{\alpha - 1}}^{\frac{\lambda}{\mu}},
$
and strictly increasing for
$
    x > \prn*{\frac{\alpha}{\alpha - 1}}^{\frac{\lambda}{\mu}}.
$
Moreover,
\begin{align}
    g(1) &= -1 - \frac{\mu}{\lambda(\alpha - 1)} < 0, \\
    \lim_{x \to +\infty} g(x) &= +\infty.
\end{align}
Hence, the equation $g(x)=0$ has a unique solution for $x>1$. Denote this solution by $x_{\sol}$. Then $g(x) < 0$ for $1 < x < x_{\sol}$, $g(x_{\sol}) = 0$, and $g(x) > 0$ for $x > x_{\sol}$. 
Thus, the monotonicity table for $\log f_T(t)$ can be written as follows using $t_{\sol} \coloneqq T - \frac{x_{\sol}^{\alpha - 1}-1}{\lambda}$:

(i) If $t_{\sol} > 0$, 

\vspace{0.5\baselineskip}
\begin{tblr}{
    width = {0.95\linewidth},
    hline{1,Z} = {0.08em},
    hline{2,Y},
    vlines,
    columns = { mode = dmath, halign = c, co = 1 },
  }
  t               & 0 & \cdots           & t_{\sol}  & \cdots & T         \\
  \frac{d \log f_T(t)}{dt}              &   & +                & 0  & -      &          \\
  \log f_T(t)    &   & \nearrow &  & \searrow &  \\
\end{tblr}
\vspace{0.5\baselineskip}

(ii) If $t_{\sol} \leq 0$, 

\vspace{0.7\baselineskip}
\begin{tblr}{
    width = {0.9\linewidth},
    hline{1,Z} = {0.08em},
    hline{2,Y},
    vlines,
    columns = { mode = dmath, halign = c, co = 1 },
  }
  t       & 0 & \cdots & T         \\
  \frac{d \log f_T(t)}{dt}   &   &    -   &           \\
  \log f_T(t)    &   & \searrow &     \\
\end{tblr}
\vspace{0.5\baselineskip}

The claim follows immediately from these monotonicity tables.

\end{proof}
\subsection{Proof of \cref{thm:x convergence}}
\begin{proof}

Because $x_{\sol}$ is the solution of \cref{eq:x_0 equation}, we have 
\begin{align}
(\alpha-1)x_{\sol}^{\frac{\mu}{\lambda}+\alpha-1}
=
\alpha\left(\frac{\mu}{\lambda(\alpha-1)}+1\right)x_{\sol}^{\alpha-1}
-\frac{\mu}{\lambda}.
\end{align}
Dividing both sides by \(x_{\sol}^{\alpha-1}\) yields
\begin{align}
(\alpha-1)x_{\sol}^{\frac{\mu}{\lambda}}
&=
\alpha\left(\frac{\mu}{\lambda(\alpha-1)}+1\right)
-\frac{\mu}{\lambda}x_{\sol}^{-(\alpha-1)} \\
&= \frac{\mu}{\lambda(\alpha-1)}+\alpha
+
\frac{\mu}{\lambda}\cdot
\frac{x_{\sol}^{\alpha-1}-1}{x_{\sol}^{\alpha-1}}.
\end{align}
Hence,
\begin{align}
\frac{\mu}{\lambda}\cdot
\frac{x_{\sol}^{\alpha-1}-1}{x_{\sol}^{\alpha-1}}
=
(\alpha-1)x_{\sol}^{\frac{\mu}{\lambda}}
-\frac{\mu}{\lambda(\alpha-1)}-\alpha.
\end{align}
Since \(x_{\sol}>1\), the left-hand side is positive, and therefore
\begin{align}
(\alpha-1)x_{\sol}^{\frac{\mu}{\lambda}}
\ge
\frac{\mu}{\lambda(\alpha-1)}+\alpha.
\end{align}
It follows that
\begin{align}
x_{\sol}^{\frac{\mu}{\lambda}}
\ge
\frac{\mu}{(\alpha-1)^2\lambda}+\frac{\alpha}{\alpha-1}.
\end{align}
Taking logarithms, we obtain
\begin{align}
\log x_{\sol}
\ge
\frac{\lambda}{\mu}
\log\left(
\frac{\mu}{(\alpha-1)^2\lambda}+\frac{\alpha}{\alpha-1}
\right).
\end{align}
Since \(x_{\sol}>1\) and \(\alpha>1\), we have \((\alpha-1)\log x_{\sol}>0\), and thus
\begin{align}
x_{\sol}^{\alpha-1}-1
=
e^{(\alpha-1)\log x_{\sol}}-1
\ge
(\alpha-1)\log x_{\sol},
\end{align}
where we used the elementary inequality \(e^t-1\ge t\) for \(t\ge 0\). Combining this with the previous bound gives
\begin{align}
\frac{x_{\sol}^{\alpha-1}-1}{\lambda}
\ge
\frac{\alpha-1}{\mu}
\log\left(
\frac{\mu}{(\alpha-1)^2\lambda}+\frac{\alpha}{\alpha-1}
\right).
\end{align}
As \(\lambda\to 0\), the quantity inside the logarithm diverges to \(+\infty\), and hence
\begin{align}
\log\left(
\frac{\mu}{(\alpha-1)^2\lambda}+\frac{\alpha}{\alpha-1}
\right)\to\infty.
\end{align}
Therefore,
\begin{align}
\frac{x_{\sol}^{\alpha-1}-1}{\lambda}\to\infty
\qquad \text{as } \lambda\to 0.
\end{align}
This completes the proof.
\end{proof}

\subsection{Proof of \cref{thm:optimal goal non-exploitatie}}
\begin{proof}
By \cref{eq:trajectory_ghyp}, the condition $x(T)=\theta$ holds if and only if $t^*=T$.
If $t_{\sol} \leq 0$, then \cref{eq:t star tsol small} implies that $t^*=T$ is equivalent to $f_T(0) \leq \frac{R}{\theta^{\alpha}}$. Therefore, $x(T)=\theta$ is maximized when $\theta = R^{\frac{1}{\alpha}} f_T(0)^{-\frac{1}{\alpha}}$.
If $t_{\sol} > 0$, then \cref{eq:t star tsol large} implies that $t^*=T$ is equivalent to $f_T(t_{\sol}) \leq \frac{R}{\theta^{\alpha}}$. Therefore, $x(T)=\theta$ is maximized when $\theta = R^{\frac{1}{\alpha}} f_T(t_{\sol})^{-\frac{1}{\alpha}}$.
\end{proof}

\subsection{Proof of \cref{thm:f_T decreasing}}
\begin{proof}
From \cref{eq:def f}, we have 
\begin{align}
f_T(0) &= \prn*{\frac{\mu + \lambda (\alpha - 1)}{(\alpha -1) \bra*{\prn*{1+\lambda T}^{\frac{\mu}{\lambda (\alpha -1)}+1}-1}}}^{\alpha - 1} \prn*{1 + \lambda T}^{\frac{\mu}{\lambda}}, \\
f_T(t_{\sol}) &= \prn*{\frac{\mu + \lambda(\alpha - 1)}{(\alpha-1)\bra*{x_0^{\frac{\mu}{\lambda} + \alpha - 1}-1}}}^{\alpha-1} x_{\sol}^{\frac{\mu (\alpha -1)}{\lambda}} \cdot \\
&\exp \prn*{- \alpha \prn*{\frac{\mu}{\alpha-1} + \lambda} \int_0^{T - \frac{x_{\sol}^{\alpha - 1}-1}{\lambda}} \frac{ds}{\prn*{1 + \lambda (T-s)}^{\frac{\mu}{\lambda(\alpha - 1)}+1}-1}}. 
\end{align}

We first show that $f_T(0)$ is monotonically decreasing.
Taking the logarithm, we obtain
\begin{align}
\log f_T(0)
&=
(\alpha-1)
\log\left(
\frac{\mu+\lambda(\alpha-1)}
{(\alpha-1)\left(\left(1+\lambda T\right)^{\frac{\mu}{\lambda(\alpha-1)}+1}-1\right)}
\right)
+
\frac{\mu}{\lambda}\log(1+\lambda T).
\end{align}
Differentiating with respect to $T$, we obtain
\begin{align}
\frac{d}{dT}\log f_T(0)
&=
-(\alpha-1)
\frac{
\left(\frac{\mu}{\lambda(\alpha-1)}+1\right)\lambda
(1+\lambda T)^{\frac{\mu}{\lambda(\alpha-1)}}
}{
\left(1+\lambda T\right)^{\frac{\mu}{\lambda(\alpha-1)}+1}-1
}
+
\frac{\mu}{1+\lambda T} \\
&= \frac{
-\mu
-\lambda(\alpha-1)
\left(1+\lambda T\right)^{\frac{\mu}{\lambda(\alpha-1)}+1}
}{
(1+\lambda T)\left(\left(1+\lambda T\right)^{\frac{\mu}{\lambda(\alpha-1)}+1}-1\right)
}.
\end{align}
Since the numerator is strictly negative and the denominator is strictly positive, we obtain
\begin{align}
\frac{d}{dT}\log f_T(0)<0, 
\end{align}
and so $f_T(0)$ is strictly decreasing in $T$.

We next show that $f_T(t_{\sol})$ is strictly decreasing in $T$.
It suffices to show that
\begin{align}
I(T) \coloneqq \int_0^{T - \frac{x_{\sol}^{\alpha - 1}-1}{\lambda}}
\frac{ds}{\left(1 + \lambda (T-s)\right)^{\frac{\mu}{\lambda(\alpha - 1)}+1}-1}
\end{align}
is strictly increasing in $T$.
By setting $u=T-s$, we obtain
\begin{align}
I(T)
&=
\int_{\frac{x_{\sol}^{\alpha - 1}-1}{\lambda}}^{T}
\frac{du}{\left(1 + \lambda u\right)^{\frac{\mu}{\lambda(\alpha - 1)}+1}-1}. \\
I'(T)
&=
\frac{1}{\left(1 + \lambda T\right)^{\frac{\mu}{\lambda(\alpha - 1)}+1}-1}.
\end{align}
Since $\lambda>0$, $\mu>0$, and $\alpha>1$, the denominator is strictly positive, and thus
$
I'(T)>0.
$
Therefore, $I(T)$ is strictly increasing in $T$.
\end{proof}

\subsection{Proof of \cref{thm:optimal goal exploitatie}}
\begin{proof}
If $t_{\sol} \leq 0$, then \cref{eq:t star tsol small} implies that $x(T)=0$ when $\theta > R^{\frac{1}{\alpha}} f_T(0)^{-\frac{1}{\alpha}}$, whereas $x(T)=\theta$ when $\theta \leq R^{\frac{1}{\alpha}} f_T(0)^{-\frac{1}{\alpha}}$. Therefore, the optimal choice is $\theta = R^{\frac{1}{\alpha}} f_T(0)^{-\frac{1}{\alpha}}$.

Next, consider the case $t_{\sol} > 0$.
If $\theta > R^{\frac{1}{\alpha}} f_T(0)^{-\frac{1}{\alpha}}$, then \cref{eq:t star tsol large} implies that $x(T)=0$.

If
$R^{\frac{1}{\alpha}} f_T(t_{\sol})^{-\frac{1}{\alpha}} < \theta \leq R^{\frac{1}{\alpha}} f_T(0)^{-\frac{1}{\alpha}}$,
then $t^*=t'$, where $t'$ is defined in \cref{eq:t' def}, and hence
\begin{align}
    x(T)
    &= \theta \prn*{1 - \exp \prn*{- \prn*{\frac{\mu}{\alpha-1} + \lambda} \eta \prn*{t'}}} \\
    &= R^{\frac{1}{\alpha}} g(t').
\end{align}
Note that, by choosing $\theta$ appropriately, any point in $[0,t_{\sol}]$ can arise as $t'$.

If $\theta \leq R^{\frac{1}{\alpha}} f_T(t_{\sol})^{-\frac{1}{\alpha}}$, then $x(T)=\theta$, so the optimal choice is $\theta = R^{\frac{1}{\alpha}} f_T(t_{\sol})^{-\frac{1}{\alpha}}$.

Therefore, the optimal value of $\theta$ is $R^{\frac{1}{\alpha}} f_T(t_{\sol})^{-\frac{1}{\alpha}}$ if
$
    R^{\frac{1}{\alpha}} f_T(t_{\sol})^{-\frac{1}{\alpha}}
    \geq
    R^{\frac{1}{\alpha}} g(\tau^*),
$
and is $R^{\frac{1}{\alpha}} f_T(\tau^*)^{-\frac{1}{\alpha}}$ otherwise.
\end{proof}

\subsection{Proof of \cref{thm:alpha bound}}
\begin{proof}
It suffices to show that $g(t)$ is monotonically increasing on $(0,t_{\sol}]$. Indeed, if this is the case, then for every $\tau \in [0,t_{\sol}]$,
\begin{align}
    g(\tau) \le g(t_{\sol})
    = f_T(t_{\sol})^{-\frac{1}{\alpha}} \prn*{1 - \exp \prn*{- \prn*{\frac{\mu}{\alpha-1} + \lambda} \eta \prn*{t_{\sol}}}}
    \le f_T(t_{\sol})^{-\frac{1}{\alpha}},
\end{align}
and therefore
$
    f_T(t_{\sol})^{-\frac{1}{\alpha}}
    \geq
     g(\tau^*). 
$

We now prove that $g(t)$ is monotonically increasing on $(0,t_{\sol}]$. Let
\begin{align}
    p(t) \coloneqq 1 + \lambda (T - t), \qquad
    q(t) \coloneqq \exp \prn*{- \prn*{\frac{\mu}{\alpha-1} + \lambda} \eta \prn*{t}}.
\end{align}
Then
\begin{align}
    \log g(t) = -\frac{1}{\alpha} \log f_T(t) + \log \prn*{1 - q(t)},
\end{align}
so
\begin{align}
    \prn*{\log g(t)}' = -\frac{1}{\alpha} \prn*{\log f_T(t)}' - \frac{q'(t)}{1 - q(t)}.
\end{align}
By \cref{eq:diff log f},
\begin{align}
    \prn*{\log f_T(t)}'
    &= \lambda \cdot \frac{(\alpha - 1) p(t)^{\beta} -\alpha \beta p(t) + (\alpha - 1) (\beta - 1)}{p(t) \prn*{p(t)^{\beta}-1}},
\end{align}
where $\beta \coloneqq \frac{\mu}{\lambda (\alpha - 1)} + 1$. Also,
\begin{align}
    q'(t)
    = - \lambda \prn*{\frac{\mu}{\lambda (\alpha - 1)} + 1} \frac{q(t)}{p(t)^{\frac{\mu}{\lambda (\alpha - 1)}+1}-1}
    = - \frac{\lambda \beta q(t)}{p(t)^{\beta}-1}.
\end{align}
Substituting these expressions gives
\begin{align}
    \prn*{\log g(t)}'
    &= \frac{\lambda}{\alpha \prn*{p(t)^{\beta}-1}}
       \prn*{-(\alpha - 1) p(t)^{\beta - 1} - \frac{(\alpha - 1)(\beta - 1)}{p(t)} + \frac{\alpha \beta}{1 - q(t)}}.
\end{align}
It is enough to show that
\begin{align} \label{eq:def h}
    h(t) \coloneqq -(\alpha - 1) p(t)^{\beta - 1} - \frac{(\alpha - 1)(\beta - 1)}{p(t)} + \frac{\alpha \beta}{1 - q(t)}.
\end{align}
is non-negative on $(0,t_{\sol}]$.

To this end, we first prove the following lemma.

\begin{lemma} \label{thm:P(u)>0}
If $1 < \alpha \leq 1 + \frac{\mu + \sqrt{4 \lambda \mu + \mu^2}}{2 \lambda}$, then for
\begin{align}
P(u) &\coloneqq (\alpha \beta - 2 \alpha + 1) u^{2 \beta} + \alpha \beta u^{\beta + 1} -2 (2 \alpha - 1) (\beta - 1) u^{\beta} \\
& \quad + \alpha \beta (\beta - 1) u - (\beta - 1)(\alpha \beta - 2 \alpha - \beta + 1),
\end{align}
we have $P(u) > 0$ for every $u \ge 1$.
\end{lemma}

\begin{proof}
A direct calculation gives
\begin{align}
    P'(u) &= 2 \beta (\alpha \beta - 2 \alpha + 1) u^{2\beta - 1} + \alpha \beta (\beta + 1) u^{\beta} - 2 \beta (2\alpha - 1) (\beta - 1)u^{\beta - 1} + \alpha \beta (\beta - 1), \\
    P''(u) &= 2 \beta (2 \beta - 1) (\alpha \beta - 2 \alpha + 1) u^{2\beta - 2} +  \alpha \beta^2 (\beta + 1) u^{\beta - 1} - 2 \beta (2\alpha - 1) (\beta - 1)^2 u^{\beta - 2} \\
    &= \beta u^{\beta - 2} \prn*{2 (2 \beta - 1) (\alpha \beta - 2 \alpha + 1) u^{\beta} +  \alpha \beta (\beta + 1) u - 2 (2\alpha - 1) (\beta - 1)^2},
\end{align}
and
\begin{align}
    P(1) = \beta^2, \qquad
    P'(1) = 2 \beta^2, \qquad
    P''(1) = \beta^2 \prn*{\alpha (\beta - 1) + 2 \beta}.
\end{align}
Since $1 < \alpha \leq 1 + \frac{\mu + \sqrt{4 \lambda \mu + \mu^2}}{2 \lambda}$, we have $\alpha \beta - 2\alpha + 1 \ge 0$. Therefore, the function inside the parentheses in the expression for $P''(u)$ is increasing for $u \ge 1$, and since $P''(1)>0$, it follows that $P''(u)>0$ for all $u \ge 1$. Hence $P'(u)$ is increasing on $[1,\infty)$, and since $P'(1)>0$, we obtain $P'(u)>0$ for all $u \ge 1$. Finally, $P(1)>0$ implies $P(u)>0$ for all $u \ge 1$.
\end{proof}

\begin{lemma} \label{thm:h'(t)>0}
If $h(\tilde{t}) = 0$, then $h'(\tilde{t}) > 0$.
\end{lemma}

\begin{proof}
Differentiating $h(t)$ yields
\begin{align} \label{eq:diff h}
    h'(t) 
    &= \frac{(\alpha - 1) (\beta - 1)\lambda \prn*{p(t)^{\beta} - 1}}{p(t)^2} - \frac{\alpha \beta^2 \lambda q(t)}{\prn*{p(t)^{\beta} - 1} \prn*{1 - q(t)}^2}.
\end{align}
Now suppose that $h(\tilde{t})=0$. By \cref{eq:def h},
\begin{align}
    q(\tilde{t}) = \frac{(\alpha - 1)p(\tilde{t})^{\beta} + (\alpha - 1)(\beta - 1) - \alpha \beta p(\tilde{t})}{(\alpha - 1)p(\tilde{t})^{\beta} + (\alpha - 1)(\beta - 1)}.
\end{align}
Substituting this into \cref{eq:diff h}, we obtain
\begin{align}
    h'(\tilde{t})
    &= \frac{\lambda (\alpha - 1)}{\alpha p(\tilde{t})^2 \prn*{p(\tilde{t})^{\beta}-1}} P(p(\tilde{t})).
\end{align}
Since $p(\tilde{t}) \ge 1$, \cref{thm:P(u)>0} implies that $P(p(\tilde{t}))>0$, and therefore $h'(\tilde{t})>0$.
\end{proof}

We now return to the proof that $h(t)$ is nonnegative on $(0,t_{\sol}]$. Since
\begin{align}
    \lim_{t \to 0+} h(t) = +\infty,
\end{align}
we have $h(t)>0$ for sufficiently small $t>0$. Suppose for contradiction that there exists some $t \in (0,t_{\sol}]$ such that $h(t)<0$. Then, by continuity, there exists some $\tilde{t} \in (0,t_{\sol}]$ such that $h(\tilde{t})=0$. Choosing $\tilde{t}$ as the first such point after which $h$ becomes negative, we must have $h'(\tilde{t}) \le 0$, contradicting \cref{thm:h'(t)>0}. Therefore, $h(t)\ge 0$ on $(0,t_{\sol}]$.
\end{proof}

\subsection{Proof of \cref{thm:orsp fixed N}}
\begin{proof}
Since exploitative rewards are not allowed, when $R_i$ and $T_i$ are fixed at each stage $i$, \cref{thm:optimal goal non-exploitatie} implies that the optimal choice is
\begin{align}
\theta_i = R_i^{\frac{1}{\alpha}} F(T_i)^{\frac{\alpha-1}{\alpha}}. 
\end{align}
Therefore, the optimal reward-scheduling problem with fixed $N$ reduces to the following optimization problem over $\prn*{(T_i)_{i=1}^N, (R_i)_{i=1}^N}$:
\begin{align} \label{eq:optimization NTR}
    \max_{(T_i)_{i=1}^N,\,(R_i)_{i=1}^N}
    \quad & \sum_{i=1}^N R_i^{\frac{1}{\alpha}} F(T_i)^{\frac{\alpha-1}{\alpha}} \\
    \text{s.t.} \quad
    & \sum_{i=1}^N R_i = R, \quad R_i \geq 0 \quad (i = 1, \dots, N),\\
    & \sum_{i=1}^N T_i = T, \quad T_i \geq 0 \quad (i = 1, \dots, N). 
\end{align}
Furthermore, by Hölder's inequality,
\begin{align}
\sum_{i=1}^N R_i^{\frac{1}{\alpha}} F(T_i)^{\frac{\alpha-1}{\alpha}}
\leq
\prn*{\sum_{i=1}^N R_i}^{\frac{1}{\alpha}}
\prn*{\sum_{i=1}^N F(T_i)}^{\frac{\alpha - 1}{\alpha}}
=
R^{\frac{1}{\alpha}}
\prn*{\sum_{i=1}^N F(T_i)}^{\frac{\alpha - 1}{\alpha}}.
\end{align}
Since equality holds if and only if $R_i \propto F(T_i)$, the optimization problem \cref{eq:optimization NTR} can be further reduced to
\begin{align} \label{eq:optimization NT}
    \max_{(T_i)_{i=1}^N}
    \quad & \sum_{i=1}^N F(T_i)\\
    \text{s.t.} \quad 
    & \sum_{i=1}^N T_i = T, \quad T_i \geq 0 \quad (i = 1, \dots, N). 
\end{align}
To investigate the structure of the optimal solution to the optimization problem \cref{eq:optimization NT}, we use the following properties of the function $F$.
\begin{lemma} \label{thm:F concave}
    $F$ is continuously differentiable on $[0,\infty)$ and concave on $[0,\infty)$.
\end{lemma}
\begin{proof}
We first show that $F$ is continuously differentiable on $[0, \infty)$. Since $F_1$ and $F_2$ are clearly continuous and differentiable for $T \geq 0$, it suffices to show continuity and differentiability at $T = \frac{x_{\sol}^{\alpha - 1}-1}{\lambda}$, that is,
\begin{align}
    F_1 \prn*{\frac{x_{\sol}^{\alpha - 1}-1}{\lambda}} = F_2 \prn*{\frac{x_{\sol}^{\alpha - 1}-1}{\lambda}}
\end{align}
and
\begin{align}
    F_1' \prn*{\frac{x_{\sol}^{\alpha - 1}-1}{\lambda}} = F_2' \prn*{\frac{x_{\sol}^{\alpha - 1}-1}{\lambda}}.
\end{align}
From \cref{eq:def f}, we obtain
\begin{align}
    F_1(T) 
    &= \frac{\alpha-1}{\mu + \lambda(\alpha - 1)} \prn*{1 + \lambda T - \prn*{1 + \lambda T}^{-\frac{\mu}{\lambda(\alpha - 1)}}}, \\
    F_1'(T) &= \frac{\lambda(\alpha-1)}{\mu + \lambda(\alpha - 1)} \prn*{1 + \frac{\mu}{\lambda (\alpha - 1)} \prn*{1 + \lambda T}^{-\frac{\mu}{\lambda(\alpha - 1)}-1}}, \\
    F_2(T) 
    &= \frac{(\alpha-1)\prn*{{x_{\sol}^{\alpha - 1}- x_{\sol}^{-\frac{\mu}{\lambda}}}}}{\mu + \lambda(\alpha - 1)} \exp \prn*{\frac{\alpha \prn*{\mu + \lambda (\alpha - 1)}}{(\alpha - 1)^2} \eta \prn*{T - \frac{x_{\sol}^{\alpha-1}-1}{\lambda}}}, \\
    F_2'(T) &= \frac{\alpha \prn*{{x_{\sol}^{\alpha - 1}- x_{\sol}^{-\frac{\mu}{\lambda}}}}}{(\alpha - 1) \prn*{\prn*{1 + \lambda T}^{\frac{\mu}{\lambda(\alpha -1)} + 1}-1}} \exp \prn*{\frac{\alpha \prn*{\mu + \lambda (\alpha - 1)}}{(\alpha - 1)^2} \eta \prn*{T - \frac{x_{\sol}^{\alpha-1}-1}{\lambda}}}.
\end{align}

Now,
\begin{align}
    F_1 \prn*{\frac{x_{\sol}^{\alpha - 1}-1}{\lambda}} &= \frac{\alpha-1}{\mu + \lambda(\alpha - 1)} \prn*{x_{\sol}^{\alpha - 1} - x_{\sol}^{-\frac{\mu}{\lambda}}},\\
    F_2 \prn*{\frac{x_{\sol}^{\alpha - 1}-1}{\lambda}} &= \frac{\alpha-1}{\mu + \lambda(\alpha - 1)} \prn*{x_{\sol}^{\alpha - 1} - x_{\sol}^{-\frac{\mu}{\lambda}}},
\end{align}
so
\begin{align}
    F_1 \prn*{\frac{x_{\sol}^{\alpha - 1}-1}{\lambda}} = F_2 \prn*{\frac{x_{\sol}^{\alpha - 1}-1}{\lambda}}.
\end{align}
Moreover,
\begin{align}
    F_1' \prn*{\frac{x_{\sol}^{\alpha - 1}-1}{\lambda}} &= \frac{\lambda(\alpha-1)}{\mu + \lambda(\alpha - 1)} \prn*{1 + \frac{\mu}{\lambda (\alpha - 1)} x_{\sol}^{-\frac{\mu}{\lambda} - \alpha + 1}},\\
    F_2' \prn*{\frac{x_{\sol}^{\alpha - 1}-1}{\lambda}} &= \frac{\alpha}{\alpha - 1} x_{\sol}^{-\frac{\mu}{\lambda}},
\end{align}
and therefore
\begin{align}
    &F_1' \prn*{\frac{x_{\sol}^{\alpha - 1}-1}{\lambda}} - F_2' \prn*{\frac{x_{\sol}^{\alpha - 1}-1}{\lambda}} \\
    &= \frac{\lambda}{\mu + \lambda (\alpha - 1)} x_{\sol}^{-\frac{\mu}{\lambda} - \alpha + 1} \prn*{(\alpha - 1) x_{\sol}^{\frac{\mu}{\lambda} + \alpha - 1} - \alpha \prn*{\frac{\mu}{\lambda} \frac{1}{\alpha - 1} + 1}x_{\sol}^{\alpha - 1} + \frac{\mu}{\lambda}} \\
    &= 0.
\end{align}
Here we used the fact that $x_{\sol}$ is a solution to \cref{eq:x_0 equation}. 
This indicates
\begin{align}
    F_1' \prn*{\frac{x_{\sol}^{\alpha - 1}-1}{\lambda}} = F_2' \prn*{\frac{x_{\sol}^{\alpha - 1}-1}{\lambda}}.
\end{align}
and that $F$ is continuously differentiable on $[0, \infty)$.

Next, we show that $F$ is concave. For $F_1(T)$,
\begin{align}
    F_1''(T) = - \frac{\mu}{\alpha - 1} \prn*{1 + \lambda T}^{-\frac{\mu}{\lambda(\alpha - 1)}-2},
\end{align}
and hence $F_1(T)$ is concave for $T \geq 0$. Next, for $F_2(T)$, we compute
\begin{align}
    F_2''(T) 
    &= \frac{\alpha \prn*{x_{\sol}^{\alpha - 1}- x_{\sol}^{-\frac{\mu}{\lambda}}} \prn*{\mu + \lambda (\alpha - 1)}}{(\alpha - 1)^2} \exp \prn*{\frac{\alpha \prn*{\mu + \lambda (\alpha - 1)}}{(\alpha - 1)^2} \eta \prn*{T - \frac{x_{\sol}^{\alpha-1}-1}{\lambda}}} \\
    &\qquad \cdot \frac{\frac{\alpha}{\alpha - 1} - (1 + \lambda T)^{\frac{\mu}{\lambda(\alpha - 1)}}}{\prn*{\prn*{1 + \lambda T}^{\frac{\mu}{\lambda(\alpha -1)} + 1}-1}^2}.
\end{align}
Thus, it suffices to show that
\begin{align}
    \frac{\alpha}{\alpha - 1} - (1 + \lambda T)^{\frac{\mu}{\lambda(\alpha - 1)}} < 0.
\end{align}
When $T > \frac{x_{\sol}^{\alpha - 1} - 1}{\lambda}$, we have $(1 + \lambda T)^{\frac{1}{\alpha - 1}} > x_{\sol}$. Moreover, since
\begin{align}
    (\alpha - 1) x_{\sol}^{\frac{\mu}{\lambda} + \alpha - 1} - \alpha \prn*{\frac{\mu}{\lambda} \frac{1}{\alpha - 1} + 1}x_{\sol}^{\alpha - 1} + \frac{\mu}{\lambda} = 0,
\end{align}
it follows that
\begin{align}
    x_{\sol}^{\frac{\mu}{\lambda}} = \frac{\alpha \prn*{\mu + \lambda (\alpha - 1)}}{\lambda(\alpha - 1)^2} - \frac{\mu}{\lambda (\alpha - 1)x_{\sol}^{\alpha - 1}}.
\end{align}
Therefore,
\begin{align}
     \frac{\alpha}{\alpha - 1} - (1 + \lambda T)^{\frac{\mu}{\lambda(\alpha - 1)}} &< \frac{\alpha}{\alpha - 1} - x_{\sol}^{\frac{\mu}{\lambda}} \\
     &= \frac{\alpha}{\alpha - 1} - \frac{\alpha \prn*{\mu + \lambda (\alpha - 1)}}{\lambda(\alpha - 1)^2} + \frac{\mu}{\lambda (\alpha - 1)x_{\sol}^{\alpha - 1}}  \\
     &= -\frac{\alpha \mu}{\lambda (\alpha - 1)^2} + \frac{\mu}{\lambda (\alpha - 1)x_{\sol}^{\alpha - 1}} \\
     &= \frac{\mu}{\lambda (\alpha - 1)^2 x_{\sol}^{\alpha - 1}} \prn*{\alpha \prn*{1 - x_{\sol}^{\alpha - 1}} - 1} < 0.
\end{align}
This proves that $F_2(T)$ is concave when $T > \frac{x_{\sol}^{\alpha - 1} - 1}{\lambda}$.
\end{proof}
By the concavity of $F$ established in \cref{thm:F concave} and Jensen's inequality,
\begin{align}
    \sum_{i=1}^N F(T_i)
    = N \sum_{i=1}^N \frac{1}{N} F(T_i)
    \leq N \cdot F \prn*{\sum_{i=1}^N \frac{1}{N} T_i}
    = N \cdot F \prn*{\frac{T}{N}}.
\end{align}
Equality holds if and only if $T_i = \frac{T}{N}$ for all $i$.
\end{proof}

\subsection{Proof of \cref{thm:G(N)}}
\begin{proof}
Let $G_1(N) \coloneqq N \cdot F_1 \prn*{\frac{T}{N}}$ and $G_2(N) \coloneqq N \cdot F_2 \prn*{\frac{T}{N}}$. It suffices to show that $G_1$ is monotonically increasing when $\frac{T}{N} \leq \frac{x_{\sol}^{\alpha - 1}-1}{\lambda}$, and that $G_2$ is monotonically increasing when $\frac{T}{N} > \frac{x_{\sol}^{\alpha - 1}-1}{\lambda}$.

We first consider $G_1$. We have
\begin{align}
    G_1(N) &= \frac{\alpha - 1}{\mu + \lambda (\alpha - 1)} \prn*{N + \lambda T - N \prn*{1 + \frac{\lambda T}{N}}^{-\frac{\mu}{\lambda (\alpha - 1)}}}, \\
    G_1'(N) &= \frac{\alpha - 1}{\mu + \lambda (\alpha - 1)} \prn*{1 - \prn*{1 + \frac{\lambda T}{N}}^{-\frac{\mu}{\lambda (\alpha - 1)}} - \frac{\mu T}{(\alpha - 1)N} \prn*{1 + \frac{\lambda T}{N}}^{-\frac{\mu}{\lambda (\alpha - 1)}-1}}.
\end{align}
Now let
\begin{align}
    x \coloneqq 1 + \frac{\lambda T}{N}, \qquad \nu \coloneqq \frac{\mu}{\lambda (\alpha - 1)}, \qquad h(x) \coloneqq x^{-\nu}.
\end{align}
Then
\begin{align}
    G_1'(N) &= \frac{\alpha - 1}{\mu + \lambda (\alpha - 1)} \prn*{1 - x^{-\nu} - \nu x^{-\nu - 1} (x - 1)} \\
    &= \frac{\alpha - 1}{\mu + \lambda (\alpha - 1)} \prn*{h(1) - h(x) + h'(x) (x - 1)} \\
    &= \frac{\alpha - 1}{\mu + \lambda (\alpha - 1)} (x - 1) \prn*{h'(x) - \frac{h(x) - h(1)}{x - 1}}.
\end{align}
By the mean value theorem, there exists some $\tilde{x}$ with $1 < \tilde{x} < x$ such that
\begin{align}
    h'(\tilde{x}) = \frac{h(x) - h(1)}{x - 1}.
\end{align}
Therefore,
\begin{align}
    G_1'(N)
    &= \frac{\alpha - 1}{\mu + \lambda (\alpha - 1)} (x - 1) \prn*{h'(x) - h'(\tilde{x})}.
\end{align}
Since $h$ is convex, we have $h'(x) - h'(\tilde{x}) > 0$. Hence $G_1'(N) > 0$, and therefore $G_1(N)$ is monotonically increasing.

Next, we consider $G_2$. We have
\begin{align}
    G_2(N) 
    &= N \cdot \frac{(\alpha-1)\prn*{{x_{\sol}^{\alpha - 1}- x_{\sol}^{-\frac{\mu}{\lambda}}}}}{\mu + \lambda(\alpha - 1)} \\
    &\cdot \exp \prn*{\frac{\alpha \prn*{\mu + \lambda (\alpha - 1)}}{(\alpha - 1)^2} \int_{0}^{\frac{T}{N} - \frac{x_{\sol}^{\alpha-1}-1}{\lambda}} \frac{ds}{\prn*{1 + \lambda \prn*{\frac{T}{N}-s}}^{\frac{\mu}{\lambda (\alpha - 1)}+1}-1}} \\
    &= N \cdot \frac{(\alpha-1)\prn*{{x_{\sol}^{\alpha - 1}- x_{\sol}^{-\frac{\mu}{\lambda}}}}}{\mu + \lambda(\alpha - 1)} \exp \prn*{\frac{\alpha \prn*{\mu + \lambda (\alpha - 1)}}{(\alpha - 1)^2} \int_{\frac{x_{\sol}^{\alpha - 1}-1}{\lambda}}^{\frac{T}{N}} \frac{ds'}{\prn*{1 + \lambda s'}^{\frac{\mu}{\lambda (\alpha - 1)}+1}-1}},
\end{align}
where in the second equality we used the change of variables $s' \coloneqq \frac{T}{N} - s$. Hence
\begin{align}
    \log G_2(N) &= \log N + \frac{\alpha \prn*{\mu + \lambda (\alpha - 1)}}{(\alpha - 1)^2} \int_{\frac{x_{\sol}^{\alpha - 1}-1}{\lambda}}^{\frac{T}{N}} \frac{ds'}{\prn*{1 + \lambda s'}^{\frac{\mu}{\lambda (\alpha - 1)}+1}-1} + \mathrm{const.},
\end{align}
and therefore
\begin{align}
    (\log G_2(N))' &= \frac{1}{N} + \frac{\alpha \prn*{\mu + \lambda (\alpha - 1)}}{(\alpha - 1)^2} \prn*{-\frac{T}{N^2}} \frac{1}{\prn*{1 + \frac{\lambda T}{N}}^{\frac{\mu}{\lambda (\alpha - 1)}+1}-1} \\
    &= \frac{(\alpha - 1)^2 N \bra*{ \prn*{1 + \frac{\lambda T}{N}}^{\frac{\mu}{\lambda (\alpha - 1)}+1}-1} - \alpha \prn*{\mu + \lambda (\alpha - 1)}T}{N^2 (\alpha - 1)^2 \prn*{\prn*{1 + \frac{\lambda T}{N}}^{\frac{\mu}{\lambda (\alpha - 1)}+1}-1}}.
\end{align}
Thus, defining
\begin{align}
    g(N) \coloneqq (\alpha - 1)^2 N \bra*{ \prn*{1 + \frac{\lambda T}{N}}^{\frac{\mu}{\lambda (\alpha - 1)}+1}-1} - \alpha \prn*{\mu + \lambda (\alpha - 1)}T,
\end{align}
it suffices to show that $g(N) \geq 0$ when
\begin{align}
    \frac{T}{N} > \frac{x_{\sol}^{\alpha - 1}-1}{\lambda}
    \quad \Leftrightarrow \quad
    0 < N < \frac{\lambda T}{x_{\sol}^{\alpha - 1} - 1}.
\end{align}
Let
$
    x \coloneqq \prn*{1 + \frac{\lambda T}{N}}^{\frac{1}{\alpha - 1}}.
$
Then
\begin{align}
    g'(N) &= (\alpha - 1) \prn*{-\frac{\mu}{\lambda} x^{\frac{\mu}{\lambda} + \alpha - 1} + \prn*{\frac{\mu}{\lambda}+ \alpha - 1} x^{\frac{\mu}{\lambda}}- \alpha + 1}.
\end{align}
Define
\begin{align}
    \tilde{g}(x) \coloneqq -\frac{\mu}{\lambda} x^{\frac{\mu}{\lambda} + \alpha - 1} + \prn*{\frac{\mu}{\lambda}+ \alpha - 1} x^{\frac{\mu}{\lambda}}- \alpha + 1.
\end{align}
Then
\begin{align}
    \tilde{g}'(x) = \frac{\mu}{\lambda} \prn*{\frac{\mu}{\lambda} + \alpha - 1} x^{\frac{\mu}{\lambda}-1} \prn*{1 - x^{\alpha - 1}} < 0,
\end{align}
so $\tilde{g}(x)$ is strictly decreasing. Since $\tilde{g}(1) = 0$ and $x \geq x_{\sol} > 1$, it follows that $\tilde{g}(x) < 0$. Therefore, $g'(N) < 0$, and hence $g(N)$ is strictly decreasing on $0 < N < \frac{\lambda T}{x_{\sol}^{\alpha - 1} - 1}$.
Because
\begin{align}
    g\prn*{\frac{\lambda T}{x_{\sol}^{\alpha - 1} - 1}} &= (\alpha - 1)^2 \frac{\lambda T}{x_{\sol}^{\alpha - 1}-1} \prn*{x_{\sol}^{\frac{\mu}{\lambda} + \alpha - 1} - 1} - \alpha \prn*{\mu + \lambda (\alpha - 1)}T \\
    &= \frac{(\alpha - 1) \lambda T}{x_{\sol}^{\alpha-1} - 1} \\
    &\cdot \bra*{\prn*{(\alpha - 1) x_{\sol}^{\frac{\mu}{\lambda} + \alpha - 1} - \alpha \prn*{\frac{\mu}{\lambda} \frac{1}{\alpha - 1} + 1}x_{\sol}^{\alpha - 1} + \frac{\mu}{\lambda}} + \frac{\mu + \lambda (\alpha - 1)}{\lambda (\alpha - 1)}} \\
    &= \frac{\prn*{\mu + \lambda (\alpha - 1)} T}{x_{\sol}^{\alpha-1} - 1} > 0, 
\end{align}
we conclude that $g(N) > 0$ for all
$
    0 < N < \frac{\lambda T}{x_{\sol}^{\alpha - 1} - 1}.
$
This proves that $G_2(N)$ is monotonically increasing on $0 < N \leq \frac{\lambda T}{x_{\sol}^{\alpha - 1} - 1}$.

Finally,
\begin{align}
\lim_{N \to \infty} G(N) &= R^{\frac{1}{\alpha}} \prn*{\lim_{N \to \infty} N \cdot F \prn*{\frac{T}{N}}}^{\frac{\alpha - 1}{\alpha}} \\
&= R^{\frac{1}{\alpha}} \prn*{\lim_{N \to \infty} G_1(N)}^{\frac{\alpha - 1}{\alpha}},
\end{align}
and
\begin{align}
    \lim_{N \to \infty} G_1(N)
    &= \frac{\alpha - 1}{\mu + \lambda (\alpha - 1)} \lim_{N \to \infty} \prn*{N + \lambda T - N \prn*{1 + \frac{\lambda T}{N}}^{-\frac{\mu}{\lambda (\alpha - 1)}}} \\
    &= \frac{\alpha - 1}{\mu + \lambda (\alpha - 1)} \lim_{N \to \infty} \prn*{N + \lambda T - N \prn*{\sum_{k=0}^{\infty} \binom{-\frac{\mu}{\lambda (\alpha - 1)}}{k} \prn*{\frac{\lambda T}{N}}^k}} \\
    &= \frac{\alpha - 1}{\mu + \lambda (\alpha - 1)} \lim_{N \to \infty} \prn*{N + \lambda T - N \prn*{1 - \frac{\mu T}{(\alpha - 1)N} + O(N^{-2})}} \\
    &= \frac{\alpha - 1}{\mu + \lambda (\alpha - 1)} \cdot \frac{\mu + \lambda (\alpha - 1)}{\alpha - 1} T \\
    &= T.
\end{align}
In the above derivation, we used the generalized binomial theorem \cite{graham1994concrete}:
\begin{align}
    (1 + z)^r = \sum_{k=0}^{\infty} \binom{r}{k} z^k
\end{align}
for $\abs{z} < 1$. 
Therefore,
$
    \lim_{N \to \infty} G(N) = R^{\frac{1}{\alpha}} T^{\frac{\alpha - 1}{\alpha}}.
$
\end{proof}

\section*{Declaration of Generative AI Use}
The authors used generative AI tools to assist with language editing and improving the clarity of the manuscript. The authors reviewed and edited all AI-assisted outputs and take full responsibility for the content of the paper.

\bibliographystyle{elsarticle-harv} 
\bibliography{main}

\end{document}